\documentclass[onecolumn,12pt,draftclsnofoot]{IEEEtran}
\IEEEoverridecommandlockouts
\usepackage{caption}
\usepackage{subcaption}
\usepackage{cite}
\usepackage{amsmath,amssymb,amsfonts}
\usepackage{algorithmic}
\usepackage{hyperref}
\hypersetup{bookmarksopen}
\usepackage{booktabs}
\usepackage{graphicx}
\usepackage{textcomp}
\usepackage{array,tabularx, makecell, multirow}
\usepackage{xcolor}
\usepackage{diagbox,multirow,lineno}
\usepackage{amsthm}
{
	\newtheorem{prob}{Problem}
	\newtheorem{corollary}{Corollary}
	
	\newtheorem{definition}{Definition}
	\newtheorem{lemma}{Lemma}
	\newtheorem{theorem}{Theorem}
	
}
\def\BibTeX{{\rm B\kern-.05em{\sc i\kern-.025em b}\kern-.08em
    T\kern-.1667em\lower.7ex\hbox{E}\kern-.125emX}}
\begin{document}

\title{Minimizing the Age of Information of Cognitive Radio-Based IoT Systems Under A Collision Constraint\\
}

\author{\IEEEauthorblockN{Qian Wang, He Chen, Yifan Gu, Yonghui Li, Branka Vucetic}
\thanks{Q.Wang is with  School of Electrical
	and Information Engineering, The University of Sydney, Sydney, NSW
	2006, Australia and Department of Information Engineering, The Chinese University of Hong Kong, Hong Kong SAR, China. The work is done when she is a visiting student at CUHK (email:qian.wang2@sydney.edu.au).}%
\thanks{H. Chen is with Department of Information Engineering, The Chinese University of Hong Kong, Hong Kong SAR, China (email: he.chen@ie.cuhk.edu.hk).}%
\thanks{Y. Gu, Y. Li and B. Vucetic are with School of Electrical
	and Information Engineering, The University of Sydney, Sydney, NSW
	2006, Australia (email: \{yifan.gu, yonghui.
		li,branka.vucetic\}@sydney.edu.au).}%

}


\maketitle

\begin{abstract}
This paper considers a cognitive radio-based IoT monitoring system, consisting of an IoT device that aims to update its measurement to a destination using cognitive radio technique. Specifically, the IoT device as a secondary user (SIoT), seeks and exploits the spectrum opportunities of the licensed band vacated by its primary user (PU) to deliver status updates without causing visible effects to the licensed operation. In this context, the SIoT should carefully make use of the licensed band and schedule when to transmit to maintain the timeliness of the status update. The timeliness of the status update characterizes how the destination knows the latest information of the SIoT. We adopt a recent metric, \textit{Age of Information} (AoI), to characterize the timeliness of the status update of the SIoT. We aim to minimize the long-term average AoI of the SIoT while satisfying the collision constraint imposed by the PU by formulating a constrained Markov decision process (CMDP) problem. We first prove the existence of optimal stationary policy of the CMDP problem. The optimal stationary policy (termed age-optimal policy) is shown to be a randomized simple policy that randomizes between two deterministic policies with a fixed probability. We prove that the two deterministic policies have a threshold structure and further derive the closed-form expression of average AoI and collision probability for the deterministic threshold-structured policy by conducting Markov Chain analysis. The analytical expression offers an efficient way to calculate the threshold and randomization probability to form the age-optimal policy. For comparison, we also consider the throughput maximization policy (termed throughput-optimal policy) and analyze the average AoI performance under the throughput-optimal policy in the considered system. Numerical simulations show the superiority of the derived age-optimal policy over the throughput-optimal policy. We also unveil the impacts of various system parameters on the corresponding optimal policy and the resultant average AoI.
\end{abstract}
%
\section{Introduction}
Recent years have witnessed extensive interests in the Internet of Things (IoT) from both academia and industry\cite{perera2015emerging}. Compared to wired solutions, wireless IoT networks have various advantages, such as low cost, easy deployment and maintenance, and support of mobility \cite{khan2017cognitive}. However, spectrum scarcity and low utilization of licensed wireless spectrum remain the major challenges for large scale IoT deployment due to a large number of interconnected IoT devices and the mutual interference among these devices. To address these challenges, cognitive radio (CR) represents a promising solution. In CR systems, a secondary user (SU) is allowed to access the spectrum licensed by the primary user (PU) with invisible effect to the PU. Incorporating the CR technique into IoT networks helps to effectively improve the spectrum utilization efficiency and alleviate the interference, the design and optimization of CR-based IoT (CR-IoT) networks have attracted considerable interests recently, see \cite{khan2017cognitive,khan2016cognitive} and references therein.

Most existing works focus on the delay minimization (e.g.,\cite{chen2011delay,wang2010delay}) and throughput maximization (e.g.,\cite{urgaonkar2009opportunistic,huang2009optimization}) of the SU(s) subject to the minimum performance requirements of PU of the licensed frequency band. The tradeoff between queuing delay performance of the SU and interference from the SU to the PU was studied in \cite{chen2011delay,wang2010delay}. A threshold policy was introduced in \cite{chen2011delay} to minimize the average delay while satisfying a collision constraint. The performance of the threshold policy was proved to be close to that of the optimal policy. The authors in \cite{wang2010delay} analyzed the average queuing delay and interference for CR networks with multiple SUs and multiple PUs, where SUs contend the PU channels by using random access scheme. \cite{urgaonkar2009opportunistic} focused on a system with multiple parallel channels that are orthogonal in either frequency or space domain with licensed PUs. Throughput-optimal policies were developed by considering general interference and mobility of SUs by resorting to the Lyapunov Optimization technique \cite{urgaonkar2009opportunistic}. \cite{huang2009optimization} studied the CR network with one licensed band and proposed a threshold strategy that achieves close-to-optimal throughput performance.

Many emerging applications of the CR-IoT, such as smart city, smart building, health-monitoring, environment monitoring and so on\cite{perera2015emerging,khan2017cognitive,khan2016cognitive} require timely status update delivery. In these applications, the timeliness of the information could be significant. For example, in smart transportation, the sensors continuously measure and update the location information of the public transportation to users. When planning a trip, people are interested in the latest information of bus or train location, which indicates the importance of the timeliness of status update. A natural question that arises is how to characterize the timeliness of status update in CR-IoT since timely status update is fundamentally different from either delay minimization or throughput maximization \cite{kaul2012real}. In this context, a new performance metric, termed \textit{Age of Information} (AoI), has been proposed in \cite{kaul2012real}. It is defined as the time elapsed since the generation time of the latest successfully received status update at the receiver. From the definition, AoI characterizes both the generation interval between successive status updates and the network latency of each status update, while delay and throughput are not capable to characterize the timeliness of the status update. 


There have been some initial efforts on investigating the AoI of CR networks \cite{gu2019minimizing,valehi2017maximizing,leng2019age}. The authors in \cite{gu2019minimizing} focused on analyzing and comparing the average AoI performance of underlay and overlay CR access strategies with perfect spectrum sensing. \cite{valehi2017maximizing} considered the interference-free interwave CR scheme and  proposed an optimal framing and scheduling policy that maximizes the system energy efficiency by optimizing the length of status update packets subject to a constraint that the expected AoI should be bounded below a certain limit. The authors in \cite{leng2019age} have studied AoI minimization for a SU in CR energy harvesting communications. Specifically, the SU can harvest and store energy, and aims to minimize its average AoI based on the availability of its energy and the primary spectrum. In most of these work\cite{gu2019minimizing,leng2019age}, the slotted transmission was considered and strict slot synchronization between the PU and the SU was assumed. The slotted transmission is common in IoT applications, e.g., wireless networked control systems \cite{park2017wireless}. However, it is challenging to achieve perfect slot synchronization between PU and SU in practice. Furthermore, some practical networks do not have a slotted structure and work in a continuous manner, such as Bluetooth and WLAN \cite{geirhofer2007cognitive}. 

Comparing to the system with slotted and synchronized PU and SU, PUs in unsynchronized systems are more likely to suffer from interference from the SU. This is because the primary operation service may return during the transmission of the SU, even if the channel is sensed idle correctly at the beginning of the SU's transmission in practical terms. One of the primary concerns of CR is to ensure that the licensed operations are not compromised. That is, only a limited collision can be tolerated by the primary system. Hence, for CR-IoT systems with an unsynchronized PU, the SU should pay a special attention to the protection of the PU. 

Motivated by the above gap, this paper attempts to minimize the average AoI of the IoT device as a secondary user (SIoT) while carefully taking the corresponding interference to the PU into account. Specifically, we consider an unsynchronized CR-IoT monitoring system, consisting of a PU evolves as a continuous-time Markov chain and a SIoT that aims to timely update its measurement to a destination working in a slotted manner by using CR technique. The SIoT adopts ``listen-before-talk'' strategy. That is, the SIoT first senses the channel, and then decides to sample and transmit its measurement only if the channel is sensed to be idle. The collision constraint imposed by the PU requires that the collision probability of PU's transmission should be less than a specified threshold. The collision constraint guarantees the performance of the PU (licensed user) of the frequency band and is commonly used for PU protection \cite{chen2011delay,huang2009optimization}. 
We aim to find an optimal policy (termed age-optimal policy, we use ``optimal policy'' and ``age-optimal policy'' interchangeably if not specified, hereafter) for the SIoT to minimize its average AoI subject to the collision constraint of PU by scheduling when to sample and transmit its status update based on the sensing result and the instantaneous AoI value.


Different from the existing AoI optimization in non-CR networks, the evolution of AoI of the SIoT systems in our model is affected by the PU. That is, the traffic of PU will influence the AoI of CR-IoT system due to the interference between the PU and the SIoT as well as the protection required for the PU. Specifically, if the channel is sensed busy, the SIoT cannot transmit its status update. Furthermore, the SIoT should decide whether to sample and transmit its status update to achieve better average AoI performance while satisfying the collision constraint of the PU. In contrast, most of the existing AoI optimization work considered status update between source and destination with an allocated frequency band. The AoI evolution in these work is only related to the channel condition and transmission policy. The coupled AoI evolution and primary operation service as well as the collision constraint imposed by the PU make it non-trivial to find the optimal policy of average AoI minimization for the SIoT. To the best of the authors' knowledge, the optimal policy for minimizing the AoI of the SIoT under a collision constraint imposed by the PU in CR-IoT systems with unsynchronized PU and SU has not been established.

\subsection{Contributions}
An optimal schedule policy to minimize the average AoI of the SIoT in a CR-IoT monitoring system subject to a collision constraint imposed by the PU has been investigated. The main contributions of this paper can be summarized as follows.
\begin{itemize}
	\item We formulate the optimal policy design problem for the average AoI minimization of the SIoT in CR-IoT networks as a constrained Markov decision process (CMDP) with a countable state space. The collision constraint of the PU in the system is regarded as the global constraint and the AoI is regarded as the objective reward.
	\item We prove the existence of the optimal stationary policy and show that the optimal policy is a randomized simple policy that randomizes at each state between two deterministic policies with a fixed probability. We adopt the Lagrangian primal-dual method to solve the formulated CMDP problem. By analyzing the subadditive property of the Lagrangian relaxed Markov decision process (MDP), the two deterministic policies from the optimal randomized simple policy are proved to have a threshold structure. That is, the SIoT will conduct transmission only if the AoI at its destination is above a threshold and the channel is sensed to be idle. We further provide a theoretical analysis of the threshold-based policy through Markov chain analysis and derive the closed-form optimal policy based on the theoretical analysis, which significantly reduces the computation complexity of finding the optimal policy compared to the traditional relative value iteration (RVI) method.
	\item We also analyze the average AoI performance of throughput-optimal policy (the policy that maximizes the throughput of the SIoT in CR-IoT system) as a benchmark and compare both policy structure and AoI performance of throughput-optimal policy with that of the derived age-optimal policy. Extensive simulations are conducted to validate our theoretical analysis and gain more insights. Numerical simulations show the superiority of the age-optimal policy over the throughput-optimal policy. We also unveil the impacts of various system parameters on the corresponding optimal policy and the resultant average AoI. It is shown that as the PU's activity frequency increases (smaller average idle state length and busy state length in each busy-idle cycle), the optimal average AoI of the SIoT first drops significantly, and then increases slowly, given the collision constraint of the PU, when the average idle probability of the PU is fixed. The existence of the tradeoff between the PU's activity frequency and the average AoI performance of the SIoT provides an important guideline for the SIoT to choose the appropriate frequency band based on the PU activity in multiple SIoT and multiple PU systems.
\end{itemize}

\subsection{Related Work}
A considerable amount of literature on AoI optimization has been published recently \cite{sun2017update,ceran2019average,wang2018skip,wang2019minimizing,zhou2019joint,gu2019timely,wang2019age,kadota2018optimizing,hsu2018age}. Most of these work aimed to minimize the long-term average AoI by exploring sampling and updating policies for systems with single \cite{sun2017update,ceran2019average,wang2018skip,wang2019minimizing,zhou2019joint,gu2019timely,wang2019age} and multiple sources\cite{kadota2018optimizing,hsu2018age,zhou2019joint}. The optimal status update policy has been explored in \cite{sun2017update}, for a single source with the \textit{generate-at-will} status update model. It was shown that the optimal update policy outperforms the zero-wait policy in many scenarios. In \cite{ceran2019average}, the optimal policy to minimize the average AoI was studied under an average transmission probability constraint of both automatic repeat request (ARQ) and hybrid ARQ systems. Specifically, a constrained MDP problem was formulated and the structure of the optimal policy on when to conduct transmission or retransmission was derived in \cite{ceran2019average}. A single source system with the randomly generated status update was considered in \cite{wang2018skip,wang2019minimizing,gu2019timely,wang2019age}. The optimal scheduling policy to decide whether to skip the new generated status update or switch to it was studied \cite{wang2018skip}. The average AoI minimization was formulated as an MDP problem and the optimal policy was proved to be a renewal policy. \cite{wang2019minimizing} considered the system with randomly generated status updates subject to the average transmission power constraint and showed how the status generation probability influences the optimal transmission policy as well as the resultant long-term average AoI. A Truncated Automatic Repeat reQuest (TARQ) scheme was considered in \cite{gu2019timely} and the maximum allowable retransmission times was optimized by analyzing the inherent age-energy tradeoff. In \cite{wang2019age}, the packet blocklength for each status update was optimized to minimize the average AoI for different packet management schemes including non-preemption, preemption and retransmission schemes.

For multiple-source scenarios, \cite{kadota2018optimizing,hsu2018age} considered status update models with \textit{generate-at-will} and Bernoulli generation model, respectively. Compared with the single-source system, the computation of optimal scheduling policy for multiple-source systems has higher complexity due to larger dimensional state space. Hence, the complexity is another concern in policy design. The authors in \cite{kadota2018optimizing} derived the lower bound for the optimal average AoI and developed a randomized policy, a Max-Weight policy, and a Whittle’s Index policy. Whittle Index scheduling has also been applied in \cite{hsu2018age}, where the AoI optimization for multiple sources was proved to be indexable. In \cite{zhou2019joint}, both single source system and multiple source system were considered and the sampling and updating to minimize the average AoI for systems was jointly designed under an energy constraint by resorting to the MDP tool.

In most of the aforementioned work applying MDP, the existence of the structural result of the optimal policy has been proved by leveraging the special property of the problem. However, to calculate the optimal policy, RVI is always needed. But the major challenges faced by RVI is the curse of dimension in real applications. To tackle this issue, a reinforcement learning method has been adopted in \cite{zhou2019joint}. For the case of MDP with a countable state space, state truncation is usually needed. In this paper, we derived the closed-form of the optimal policy, and thus overcome the drawback of the RVI when calculating the optimal policy.

\section{System Model and Problem formulation}
We consider a CR-IoT monitoring system consisting of one SIoT that timely delivers status updates to a destination. The transmission of the status update is augmented by the spectrum sensing technique. Specifically, the SIoT is considered to exploit the spectrum opportunities of a certain spectrum band, assigned to a PU. In order to ensure the protection requirement of the PU, the SIoT needs to seek the spectrum opportunities vacated by the PU. In the following, we will present PU and SIoT model in the considered system as well as the requirement for PU protection. 
\begin{figure}[!h]
   \centerline{\includegraphics[width=0.6\textwidth]{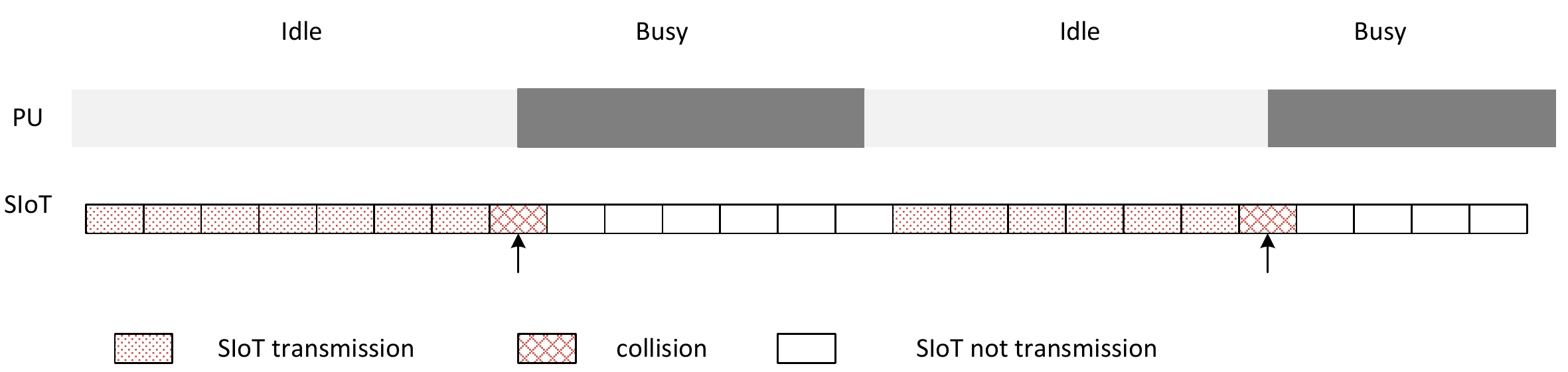}}
	\caption{PU and SIoT model illustration. }
	\label{fig0}
\end{figure}
\subsection{PU Model}
The occupancy of the channel by the PU is modeled as a two-state homogeneous continuous-time Markov chain, i.e., busy and idle states. The sojourn time for the idle state $t_I$ and that of the busy state $t_B$ are exponentially distributed with $\mathbb{E}[t_I]=\alpha^{-1}$ and $\mathbb{E}[t_B]=\beta^{-1}$, respectively. The long-term idle probability is $p_{I}=\beta/(\alpha+\beta)$. We emphasize that the PU activities are not slotted and PU has the right to access the channel at any time. A typical example is WLAN \cite{geirhofer2007cognitive}. Moreover, the assumption of exponentially distributed channel access process has been widely accepted in communication systems see, e.g., \cite{valehi2017maximizing,zhao2008opportunistic,li2011optimal}.

We also assume that the idle period of the PU is larger than the busy period, i.e, $\beta > \alpha$. This reflects the common nature of communication systems with lower spectrum utilization \cite{ghasemi2008spectrum}. As such, a spectrum access technique is demanded to overcome this shortage.


\subsection{SIoT model}

We consider a time-slotted system for the SIoT, which monitors an underlying time-varying physical process. The SIoT adopts ``listen-before-talk'' spectrum sensing scheme, and each slot consists of sensing and transmission periods. In each slot, the SIoT first senses the channel at the beginning of the time slot; then decides whether to sample and deliver the status update based on the sensing result (i.e., if the channel is sensed busy, the SIoT should be inactive, and if the spectrum is sensed idle, the SIoT will decide whether to transmit the status update or not, based on the transmission policy); and finally receives the acknowledgment from the destination for successful transmission. Different from the conventional spectrum sensing system, short packet communications is prevailing in CR-IoT systems rather than traditional long packet transmission \cite{zhang2018spectrum}. As such, the busy period of the PU could be longer than the duration of a time slot, i.e., ${\mathbb{E}}[t_B] > 1$. In addition, this phenomena makes it more flexible to design an access strategy for the SIoT. Owning to its continuous channel access process, the PU can randomly access the channel and may return in the middle of the slot during the transmission of the SIoT, even if the channel is sensed idle at the beginning of the slot. If the PU returns in the middle of the slot when the SIoT is transmitting the status update, a collision occurs and the transmission of the SIoT fails as well, as illustrated in Fig. \ref{fig0}. If the PU remains idle during the transmission of the SIoT, the SIoT's outage probability of transmission is $\phi_s$. Overall, there are two cases for transmission failure: 1) the PU returns in the middle of the slot while the SIoT is transmitting, and 2) the SIoT suffers outage during its transmission even if the PU is idle.

In this work, we assume the perfect sensing for the SIoT, with negligible small sensing time and perfect sensing outcome\footnote{The optimal policy for the scenario with imperfect sensing of SIoT has been left as a future work.}, i.e., zero false alarm rate and $100\%$ correct detection probability as in \cite{chen2011delay,gu2019minimizing}. However, as the PU might return in the middle of the time slot, the perfect spectrum sensing cannot ensure the protection of the PU transmission. In the next section, we provide the requirement for the PU protection.

\subsection{PU protection requirement}

In the considered system, there is at most one collision in each busy-idle cycle, according to the aforementioned assumptions and the continuous-time Markov chain model of the PU. The collision happens at the end of the idle period of the PU (also the beginning of a new busy period), when the PU returns in the middle of the transmission slot of the SIoT. After this slot, the SIoT continues to sense the spectrum. As we consider the perfect sensing scheme, the spectrum will be sensed busy and the SIoT will not transmit to avoid any collision in the busy period. 

To characterize the impact of the collision and further protect the PU, we introduce collision probability from the PU's perspective. We denote $\psi_p$ as the long-term average collision probability of the PU, given by
\begin{equation}
\label{eq2}
\psi_p=\lim_{T\rightarrow \infty}\frac{N_c}{N_T},
\end{equation} where $N_c$ denotes the total number of collisions that happened during $T$ slots and $N_T$ is the number of busy-idle cycles. We have a PU collision probability constraint $\psi_p \leq \eta_p$. $\eta_p$ is prior to the SIoT for transmission policy design.

\subsection{Problem Formulation}
The AoI, denoted by $\Delta(t)$ measures the timeliness of the status updates from the perspective of the receiver, i.e., the time elapsed since the generation time of the most recently successfully received status update. Mathematically, the AoI $\Delta(t)$ at time $t$ is $t-u(t)$, where $u(t)$ denotes the generation time of the latest status update at time $t$. 
\begin{figure}[!h]
	\centerline{\includegraphics[width=0.45\textwidth]{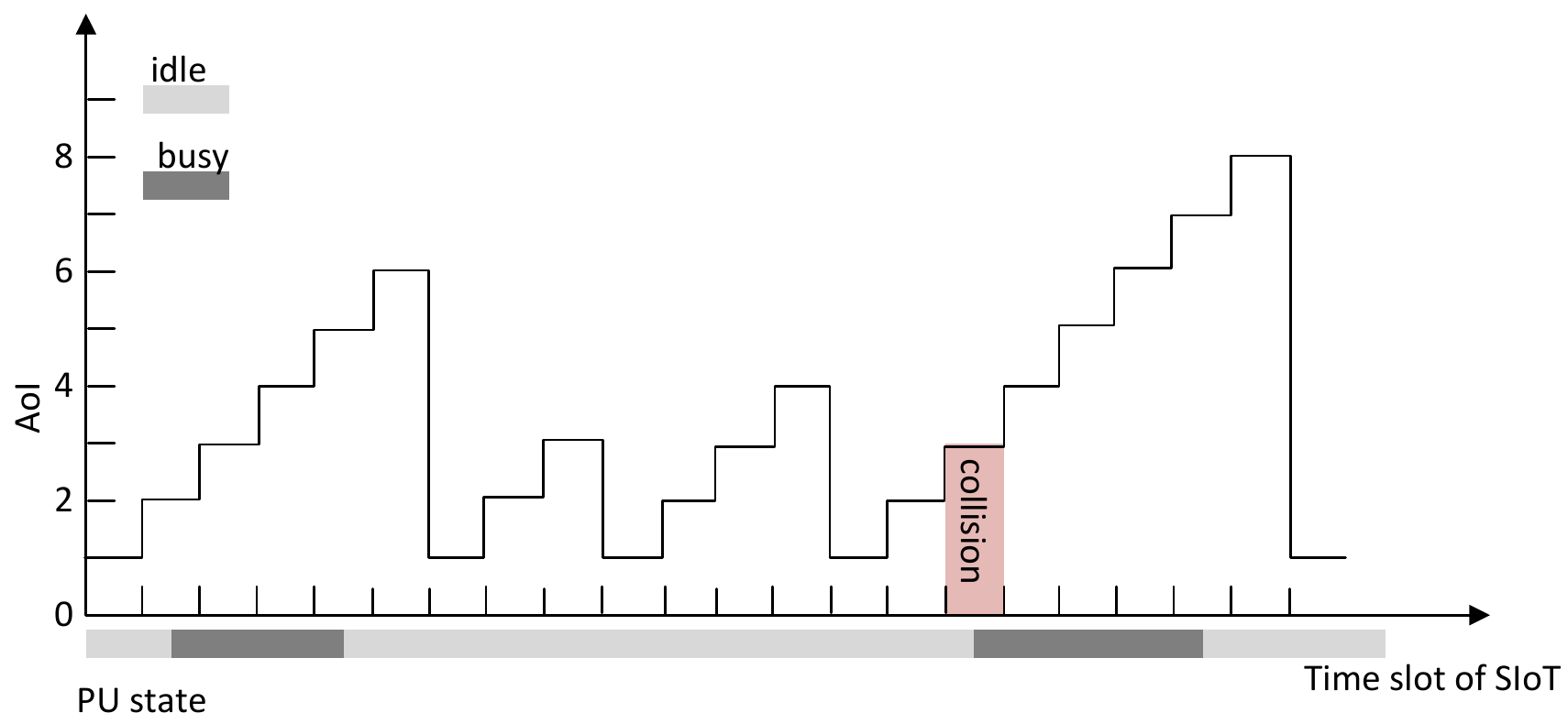}}
	\caption{The evolution of AoI for SIoT. }
	\label{figa}
\end{figure}

For the SIoT, to ensure timely status update, it should exploit transmission opportunities when the PU is idle. To have a better understanding, an example of the AoI for the SIoT is depicted in Fig. \ref{figa}. At the end of each slot, the AoI at the destination is reset to $1$ if the status update is successfully received. Otherwise, the AoI will increase by $1$. Considering the PU protection requirement, at the beginning of each slot, if the spectrum is sensed idle, the SIoT needs to decide whether to sample and transmit a new status update because the PU may return in the middle of the slot when the SIoT is transmitting the status update and suffer a collision as shown in Fig. \ref{figa}. If the spectrum is sensed busy, the SIoT should stay inactive to avoid collisions. 

We adopt the average AoI as the metric to characterize how timely the SIoT delivers its status update. The average AoI $\bar{\Delta}$ during time interval $(0,T)$ is defined as
\begin{equation}
\label{eq6}
\bar{\Delta}=\lim_{T \rightarrow \infty}\frac{\sum_{t=0}^{T}\Delta(t)dt}{T}.
\end{equation} In addition, the PU protection requirement should also be satisfied. Hence, an optimal policy is demanded to ensure timely status update of the SIoT while satisfying the collision probability constraint of the PU. 


Note that the collision probability constraint is enforced from the perspective of the PU, while the transmission policy is designed for the SIoT. The objects of the collision probability constraint and the transmission policy are inconsistent, which makes it non-trivial to find the optimal policy. Fortunately, by analyzing the PU traffic pattern and the periodic spectrum sensing model of the SIoT, we show that from the SIoT's perspective, the collision probability $\psi_s$ is given by
\begin{equation}
\label{eq2*}
\psi_s=\lim_{T\rightarrow \infty}\frac{N_c}{T}=\lim_{T\rightarrow \infty}\frac{N_c}{N_T \mathbb{E}[N_p]},
\end{equation} where $N_c$ denotes the total number of collisions happened during $T$ slots and the total number of status updates of the SIoT during $T$ slots is the same as the number of time slots $T$. $N_p$ denotes the length of a busy-idle cycle of the PU. According to the busy-idle pattern of the PU, when $T$ is approaching infinity, the second equality in \eqref{eq2*} holds.

By observing \eqref{eq2} and \eqref{eq2*}, we can find that 
\begin{equation}
\label{eq2**}
\psi_p=\psi_s \mathbb{E}[N_p]=\psi_s(1/\alpha+1/\beta).
\end{equation} 
Thanks to the relationship between the PU collision probability $\psi_p$ and the SIoT collision probability $\psi_s$, the PU collision probability constraint $\psi_p \leq \eta_p$ is the same as $\psi_s \leq \eta_s$, where $\eta_p=\eta_s\mathbb{E}[N_p]$. By this transformation, denoting $\pi$ as the transmission policy, our optimization problem can be formulated as follow
\begin{prob}
	\label{p1}
	\begin{equation}
	\begin{aligned}
	& \min_{\pi} \bar{\Delta}(\pi),\\
	& {\rm{s.t.}} \   \psi_s(\pi) \leq \eta_s.
	\end{aligned}
	\end{equation}
\end{prob}

\section{CMDP framework for optimal policy}

In this section, the optimal policy for the SIoT will be investigated under the MDP framework to minimize the average AoI of the SIoT subject to the collision probability constraint $\psi_s \leq \eta_s$. We assume that the SIoT has perfect knowledge of $\eta_s$, and PU traffic pattern, i.e., $\alpha$ and $\beta$. 

As we consider the slotted SIoT system and continuous alternating busy-idle activity pattern of the PU, we have the following transition matrix for channel state between adjacent slots \cite{papoulis2002probability},
\begin{equation}
\begin{aligned}
\label{eq1}
\Sigma=
\begin{pmatrix}
p_{II}& p_{IB}\\
p_{BI}& p_{BB}
\end{pmatrix}=\frac{1}{\alpha+\beta}\begin{pmatrix}
\beta+\alpha e^{-(\alpha+\beta)}&\alpha-\alpha e^{-(\alpha+\beta)}\\
\beta-\beta e^{-(\alpha+\beta)}& \alpha+\beta e^{-(\alpha+\beta)}
\end{pmatrix},
\end{aligned}
\end{equation} where $p_{ij}$ is the probability of channel state transiting from $i$ to $j$, $i$,$j$$\in\{I,B\}$ with $I$ and $B$ indicate idle and busy, respectively. 

\subsection{CMDP Formulation}
We are ready to formulate the problem of minimizing the average AoI of the SIoT subject to the collision constraint as a constrained MDP (CMDP) problem. An optimal policy can be obtained by using the effective tool of MDP. More specifically, the CMDP problem can be described by a tuple $\{\mathcal{S},\mathcal{A},P,r,d\}$, where

 \begin{itemize}
	\item State space $\mathcal{S}$: The state in each slot is represented by AoI and channel state (sensing result), denoted by $s_t\triangleq (\delta_t,u_t)$, where $\delta_t\in \{1,2,3,...\}$ and $u_t \in \{0,1\}$ with $0$ denoting idle and $1$ denoting busy. The state space is in two dimensions $\mathcal{S}\triangleq Z^+ \times 2$.
	\item Action space $\mathcal{A}$: If the channel is busy, then the SIoT will not transmit the status update. If the channel is idle, then the SIoT will make a decision on whether to transmit its status update or not based on its current AoI. The action at each slot $a_t \in \{0,1\}$ with $1$ denoting transmission and $0$ denoting not to transmit.
	\item Transition probabilities $P$: $P(s_{t+1}|s_t,a_t)$ is the probability of transit from state $s_t$ to $s_{t+1}$ when taking action $a_t$. Under the PU activity model and transmission model, we have
		\begin{equation}
		\label{e2}
		\begin{aligned}
		P(s_{t+1}|s_t,a_t)=P( \delta_{t+1},u_{t+1}|\delta_t,u_t,a_t)=P(\delta_{t+1}|\delta_t,u_{t+1},u_t,a_t)P(u_{t+1}|u_t).\\
		\end{aligned}
		\end{equation} To be more specific,
		\begin{equation}
		\label{e3}
		\begin{aligned}
		P(\delta_t+1,u_{t+1}|\delta_t,u_t,a_t=0)&=P(u_{t+1}|u_t),\\
		P(1,u_{t+1}=0|\delta_t,u_t=0,a_t=1)&=(1-\phi_s) e^{-\alpha},\\
		P(\delta_t+1,u_{t+1}=1|\delta_t,u_t=0,a_t=1)&=p_{IB},   \\
		P(\delta_t+1,u_{t+1}=0|\delta_t,u_t=0,a_t=1)&=p_{II}-(1-\phi_s) e^{-\alpha},
		\end{aligned}
		\end{equation} and otherwise, $P(s_{t+1}|s_t,a_t)=0$. $P(u_{t+1}|u_t)$ is given in \eqref{eq1}. $(1-\phi_s)e^{-\alpha}$ is the probability that the SIoT successfully transmits its status update. Specifically, $e^{-\alpha}$ is the probability that the channel remains idle during the transmission according to the PU model, and $1-\phi_s$ is the probability that SIoT suffers no outage.
	\item  $r: \mathcal{S} \times \mathcal{A}  \rightarrow R$ is the immediate reward based on the reward function of state-action pairs, defined by $r((\delta,u),a)=\delta$. The reward of each state-action pair is defined as the corresponding AoI of each state as we aim to minimize the average AoI.
	\item $d: \mathcal{S} \times \mathcal{A} \rightarrow R$ is the immediate cost of taking action $a$ in state $s$. Cost function of state-action pairs is $d((\delta,u),a)=a u+a(1-u)(1-e^{-\alpha})$. The cost function is defined as the average number of collision caused by action $a$ in state $(\delta,u)$. Specifically, if the action is not to sample and transmit the status, i.e., $a=0$, no collision will happen, i.e., $d((\delta,u),0)=0$. Otherwise, the average number of collisions will depend on the sensing result of channel. That is, if the sensing result is busy, i.e., $u=1$, the transmission of the SIoT will definitely lead to a collision and $d((\delta,1),1)=1$; if the sensing result is idle, i.e., $u=0$, whether the transmission of the SIoT will cause a collision to the PU depends on whether the PU will return during the transmission of the SIoT. The probability of PU keeping idle during the transmission of the SIoT is $\exp(-\alpha)$. Hence, $d((\delta,0),1)=1-\exp(-\alpha)$.
	
\end{itemize}
Given the initial state of the SIoT $s_0$, the infinite-horizon average reward of any feasible policy $\pi \in \Pi$ can be expressed as 
	\begin{equation}
	\label{e4}
	C(\pi|s_0)=\lim_{T \rightarrow \infty}\sup\frac{1}{T} \sum_{k=0}^{T}{\mathbb E}_\pi[r(s_k,a_k)|s_0].
	\end{equation} Define the infinite-horizon average cost with respect to
policy $\pi \in \Pi$ as 
	\begin{equation}
	\label{e5}
	D(\pi|s_0)=\lim_{T \rightarrow \infty}\sup\frac{1}{T} \sum_{k=0}^{T}{\mathbb E}_\pi[d(s_k,a_k)|s_0].
	\end{equation}
	
Here, ${\mathbb E}[\cdot]$ denotes the expectation with respect to dynamic PU activity and transmission outage of the SIoT under the policy $\pi$. Our objective is to find the optimal policy that minimizes the average AoI subject to the average collision probability constraint, which can be formulated as the following CMDP problem
\begin{prob}
	\label{p2}
		\begin{equation}
		\label{pro2}
		\begin{aligned}
		& \min_{\pi} C(\pi|s_0),\\
		& {\rm{s.t.}} \    D(\pi|s_0)\leq \eta_s.
		\end{aligned}
		\end{equation}
\end{prob} Considering that every time when the transmission is successful, the state of the SIoT will return to the state whose AoI is $1$ and the channel state is idle, i.e., $s=(1,0)$, we thus set $s_0=(1,0)$ as the initial state of the system.

\subsection{Optimal Policy}
Problem \ref{p2} in \eqref{pro2} is a CMDP problem with a countable but infinite state space and a finite action space. To solve 
Problem \ref{p2}, we first prove the existence of an optimal stationary policy. 

According to the definition of the collision probability, it is obvious that the CMDP problem here is feasible because we can always find a policy that satisfies the collision constraint by limiting the transmission of the SIoT. In addition, as the reward function is actually the AoI at each time slot and AoI at each slot either increases by $1$ or decreases to $1$, the reward function satisfies the following condition 
\begin{equation}
\label{e6}
\forall n >0 , {\rm the \ set} \ \{s \in S: \inf_a r(s,a) < n\}\ {\rm is \ finite}.
\end{equation} Then, from \cite[Theorem 11.7]{altman1999constrained}, Corollary \ref{c1} given below is straightforward.
\begin{corollary}
	\label{c1}
	There exists an optimal stationary policy for the CMDP given in Problem \ref{p2}.
\end{corollary}

To solve the above CMDP problem in \ref{p2}, the Lagrangian primal-dual method can be applied \cite{sennott1993constrained}. To proceed, we define
\begin{equation}
\label{e7}
L_{\lambda}(\pi|s_0)=J_{\lambda}(\pi|s_0)-\lambda\eta_s,
\end{equation} where $\lambda>0$, $J_{\lambda}(\pi|s_0)=C(\pi|s_0)+\lambda D(\pi|s_0)$. $\pi_{\lambda}^*$ denotes the optimal policy for the Lagrangian relaxed unconstrained MDP of a given $\lambda$ that achieves both $\min_{\pi}L_{\lambda}(\pi|s_0)$ and $\min_{\pi}J_{\lambda}(\pi|s_0)$, where the average reward function $r_\lambda(s,a)=r(s,a)+\lambda d(s,a).$

We subsequently show the characteristics of the optimal stationary policy by having the following theorem.
\begin{theorem}
	\label{T1}
	If $D(\pi^{*}_{\lambda=0}|s_0) \leq \eta_s$, then there exists an optimal stationary deterministic policy for the CMDP given in Problem \ref{p2}. Otherwise, if there exists $\lambda^{*}>0$ that achieves $D(\pi_{\lambda^*}|s_0)=\eta_s$, then the optimal policy for Problem \ref{p2} is the optimal policy $\pi_{\lambda^*}$ for the unconstrained problem \ref{e7}; otherwise, there exists a randomized simple policy $(\mu, \pi^{*}_{\lambda_1}, \pi^{*}_{\lambda_2})$ that is optimal for Problem \ref{p2}. The randomized simple policy $(\mu, \pi^{*}_{\lambda_1}, \pi^{*}_{\lambda_2})$ randomizes at each state $s$, choosing $\pi^{*}_{\lambda_1}(s)$ with probability $\mu$ and $\pi^{*}_{\lambda_2}(s)$ with probability $1-\mu$, $\mu\in [0,1]$ and achieving the maximum allowable collision constraint $\eta_s$. The deterministic policies $\pi^{*}_{\lambda_1}$ and $\pi^{*}_{\lambda_2}$ differ at most one state and are optimal for their corresponding unconstrained problem given in \ref{e7}.
\end{theorem}
 \begin{proof}
 	See Appendix \ref{A1}.
 \end{proof}

The following corollary can be inferred from Theorem \ref{T1}
\begin{corollary}
	\label{c2}
	There exists a constant $J_{\lambda}^{*}$, a bounded function $h_{\lambda}(\delta,u):\mathcal{S} \rightarrow R$ and a stationary and deterministic policy $\pi_{\lambda}^{*}$, that satisfies the average reward optimality equation,
	\begin{equation}
\label{e10}
J_{\lambda}^{*}+h_{\lambda}(\delta,u)=\min_{a\in A((\delta,u))} (r_\lambda ((\delta,u),a)+{\mathbb{E}}[h_{\lambda}(\hat{\delta},\hat{u})]).
\end{equation}
	 $\forall (\delta,u) \in S$, where $\pi_{\lambda}^{*}$ is the optimal policy, $J_{\lambda}^{*}$ is the optimal average reward, and $(\hat{\delta},\hat{u})$ is the next state after $(\delta,u)$ taking action $a$.
\end{corollary}

According to Corollary \ref{c2}, for a fixed $\lambda$, there exists an optimal deterministic policy for  $\min_{\pi}J_{\lambda}(\pi|s_0)$. Along with Theorem \ref{T1}, the CMDP problem can be solved in two steps: \textbf{1)} searching for the possible $\lambda$, \textbf{2)} solving the corresponding Lagrangian relaxed unconstrained MDP problem.

Furthermore, analyzing the unconstrained MDP policy makes it possible to characterize the optimal policy for the CMDP problem. In the subsequent subsection, we will analyze the structure of the optimal policy.

\subsection{Policy Structure}

In this subsection, we derive the structural results of the optimal policy for further analysis of the relationship between system parameters and the optimal policy. We first unveil the monotonicity of the optimal policy for the unconstrained MDP, in terms of the AoI $\delta$, in the following theorem 
%
%

\begin{theorem}
	\label{T2}
	The optimal policy of the unconstrained MDP that satisfies the equation in \eqref{e10} has a threshold structure, i.e.,
	\begin{equation}
	\label{eq15}
	\pi_{\lambda}^*(\delta,0)=\left\{
	\begin{array}{rcl}
	0,& &\text{if }  \delta < \Gamma\\
	1,& &\text{otherwise. }  \\
	\end{array}
	\right.
	\end{equation}
\end{theorem}

\begin{proof}
	See Appendix \ref{A2}.
\end{proof}
 

Based on the above property and Theorem \ref{T1}, the two deterministic policies $\pi^{*}_{\lambda_1}$ and $\pi^{*}_{\lambda_2}$ that make up the optimal stationary policy $(\mu, \pi^{*}_{\lambda_1}, \pi^{*}_{\lambda_2})$, have a threshold structure. In addition, the difference between the value of the two thresholds is at most $1$, because$\pi^{*}_{\lambda_1}$ and $\pi^{*}_{\lambda_2}$ differ at most $1$ state according to Theorem \ref{T1}. That is, assuming the threshold of $\pi^{*}_{\lambda_1}$ is $\Gamma$, then, that of $\pi^{*}_{\lambda_2}$ is $\Gamma+1$. Moreover, given the collision constraint $\eta_s$, the average costs of these two policies have the property that $D(\pi^{*}_{\lambda_1}|s_0)\geq\eta_s\geq D(\pi^{*}_{\lambda_2}|s_0)$. The reason consists of two parts: first, the policy with a larger threshold is less likely to conduct a transmission, in turn, the collision probability under this policy is smaller; second, if the collision probabilities of both two policies are smaller than the collision constraint, then the collision probability of the randomized simple police will also be smaller than the collision constraint. Hence, to calculate the optimal policy, we need to calculate the threshold $\Gamma$ and the randomization probability $\mu$. Comparing with the traditional RVI method, directly calculating the parameters $\Gamma$ and $\mu$ can significantly reduce the offline computation complexity. Moreover, to conduct this policy, only these parameters rather than all the state-action pairs of the optimal policy policy are needed, which further reduces the storage of the SIoT. In the next section, detailed steps to calculate the optimal policy will be elaborated.


\section{Markov Chain Analysis of Threshold Policy}
In this section, we first exploit Markov chain to analyze the relationship among the threshold policy, collision probability, the average AoI, and other system parameters. Then, based on the theoretical analysis of the threshold policy, we calculate the threshold $\Gamma$ and the randomization probability $\mu$ to form the optimal policy. 

\subsection{Construction of two dimensional Markov Chain}

Given a threshold policy with a threshold $\Gamma$, a two-dimensional Markov chain can be constructed, as displayed in Fig. \ref{fig1}. AoI and channel state are used to represent the SIoT system state. If the AoI is $1$ at the beginning of a time slot, the channel state must be idle. This is because AoI will drop to $1$ when the transmission is successful, and if the channel state is busy, the transmission will be interfered by the PU, which in turn, leads to transmission failure. The probability of state $s=(1,1)$ is $0$, hence, omitted in Fig. \ref{fig1}.
\begin{figure}[!h]
	\centerline{\includegraphics[width=0.45\textwidth]{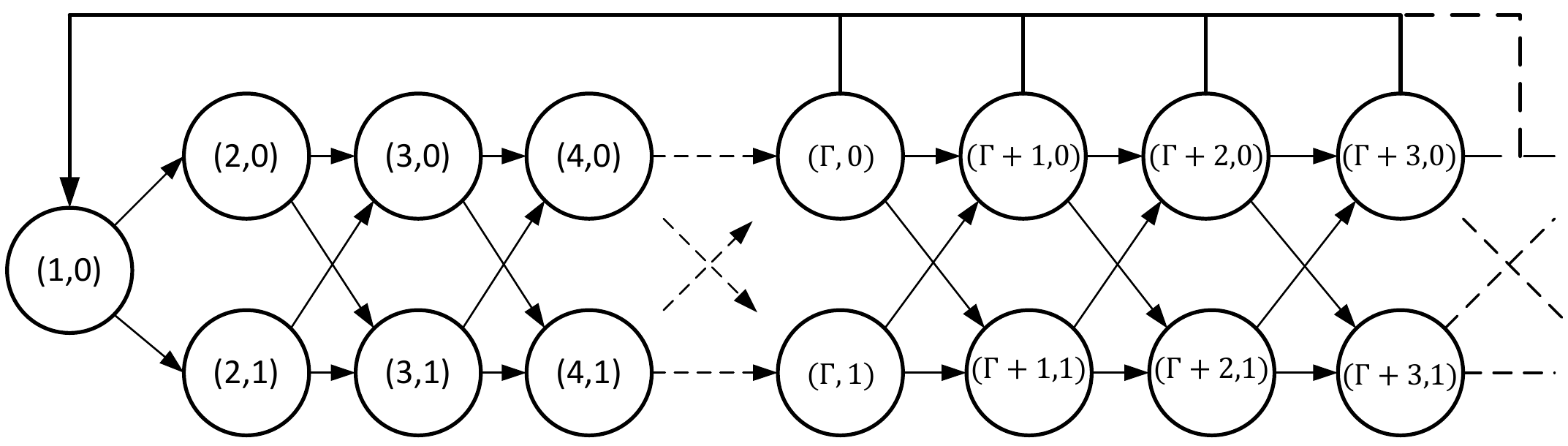}}
	\caption{Two dimensional Markov Chain under threshold policy with threshold $\Gamma$, where state $(i,j)$ indicates AoI is $i$, the channel is idle when $j=0$ and the channel is idle when $j=1$. }
	\label{fig1}
\end{figure}

We can find that the states can be divided into two parts: states whose AoI $\delta$ is smaller than $\Gamma$ and the rest states. Define $\theta=\{\theta_{(\delta,u)}\}$, $\delta \in \{1,2,...\}$ and $u \in \{0,1\}$ as the steady distribution of the Markov chain under the threshold $\Gamma$.
\subsubsection{States with $\delta < \Gamma$} For these states, the SIoT should not conduct transmission. Hence, the AoI of the next state of all these states with $\delta < \Gamma$ will increase by $1$. That is, $\sum_{u}\theta_{(\delta,u)}=\sum_{u}\theta_{(\delta+1,u)}$. Considering the transition between PU activity, we have
\begin{equation}
\label{eq16}
\left(\theta_{(\delta,0)},\theta_{(\delta,1)}\right)
\begin{pmatrix}
p_{II}& p_{IB}\\
p_{BI}& p_{BB}
\end{pmatrix}=\left(\theta_{(\delta+1,0)},\theta_{(\delta+1,1)}\right)
\end{equation}
Then, using $\theta_{(1,0)}$ and $\Gamma$, we have the following Lemma.
\begin{lemma}
\label{l2}
The steady state probabilities $\theta_{(\delta,0)}$ and $\theta_{(\delta,1)}$, for $\theta \leq \Gamma$, are
\begin{equation}
\label{e17}
\theta_{(\delta,0)}=\theta_{(1,0)}\frac{\beta+\alpha \exp(-(\alpha+\beta)(\delta-1))}{\alpha+\beta}
\end{equation}
\begin{equation}
\label{e18}
\theta_{(\delta,1)}=\theta_{(1,0)}\frac{\alpha-\alpha \exp(-(\alpha+\beta)(\delta-1))}{\alpha+\beta}.
\end{equation}
\end{lemma}
\begin{proof}
	See Appendix \ref{A3}.
\end{proof}
\subsubsection{States with $\delta \geq\Gamma$} For these states, if the channel state is idle, then the SIoT should conduct transmission. For this case, if the channel state remains idle during the transmission and the transmission suffers no outage, then system state of the SIoT will return to $s=(1,0)$; otherwise, the AoI will increase by $1$.  
If the channel state is busy, no transmission will be conducted and the AoI of the next system state will also increase by $1$.

Then, for $\delta \geq \Gamma$, we have
\begin{equation}
\label{eq18}
\left(\theta_{(\delta,0)},\theta_{(\delta,1)}\right)
\begin{pmatrix}
A& p_{IB}\\
p_{BI}& p_{BB}
\end{pmatrix}=\left(\theta_{(\delta+1,0)},\theta_{(\delta+1,1)}\right),
\end{equation} where $A=p_{II}-(1-\phi_s)e^{-\alpha}$, and
\begin{equation}
\label{eq19}
\sum_{\delta\geq\Gamma}\theta_{(\delta,0)}(1-\phi_s)e^{-\alpha}=\theta_{(1,0)}.
\end{equation} According to the relationship between states whose AoI difference is $1$, we can iteratively calculate the steady state probabilities for states with $\delta \geq\Gamma$ given in the following lemma
\begin{lemma}
	\label{l3}
	The steady state probabilities $\theta_{(\delta,0)}$ and $\theta_{(\delta,1)}$, for $\delta \geq \Gamma$, are
\begin{equation}
	\label{e19}
\theta_{(\delta,0)}=\theta_{(\Gamma,0)}\frac{p_{BI}}{m}\left(ab^{\delta-\Gamma}-cd^{\delta-\Gamma}\right)+\theta_{(\Gamma,1)}\frac{p_{BI}}{m}\left(b^{\delta-\Gamma}-d^{\delta-\Gamma}\right),
\end{equation}
\begin{equation}
\label{e20}
\theta_{(\delta,1)}=\theta_{(\Gamma,0)}\frac{p_{BI}}{m}ac\left(d^{\delta-\Gamma}-b^{\delta-\Gamma}\right)+\theta_{(\Gamma,1)}\frac{p_{BI}}{m}\left(ad^{\delta-\Gamma}-cb^{\delta-\Gamma}\right)
\end{equation}
where $m=\sqrt{(A+p_{BI})(A+p_{BI}-2)+1+4p_{BI}p_{IB}}$, $a=\frac{A+p_{BI}+m-1}{2p_{BI}}$, $b=\frac{A-p_{BI}+m+1}{2}$,$c=\frac{A+p_{BI}-m-1}{2p_{BI}}$ and $d=\frac{A-p_{BI}-m+1}{2}$.
\end{lemma}
\begin{proof}
	See Appendix \ref{A3}.
\end{proof}

As  $\sum_{u}\sum_{\delta}\theta_{(\delta,u)}=1$, based on Lemma \ref{l2} and \ref{l3}, given the threshold $\Gamma$, we have
\begin{small}
	\begin{equation}
	\label{e21}
     \theta_{(1,0)}=\frac{1}{\Gamma-1+\frac{\alpha+\beta}{\beta e^{-\alpha}(1-\phi_s)}+\frac{\alpha}{(1-\exp(-\alpha-\beta))\beta}\left(1-\exp(-(\alpha+\beta)(\Gamma-1))\right)}.
	\end{equation}
\end{small}
\subsection{Calculation of collision probability}

According to the threshold policy, the SIoT only transmits when both the channel state is idle at the beginning of the time slot and the current AoI is not smaller than $\Gamma$. In addition, the collision occurs only when the PU becomes active during the transmission of the SIoT. Thus, only these states $\{(\delta,0)\}_{\delta \geq\Gamma}$ might cause a collision. The probability of the PU staying idle during the transmission is $\exp(-\alpha)$ according to the PU model, thus $1-\exp(-\alpha)$ is the probability that a collision happens after these states. Therefore, the collision probability can be expressed as
\begin{equation}
\psi_s(\Gamma)=\sum_{\delta\geq\Gamma}\theta_{(\delta,0)}(1-\exp(-\alpha)).
\end{equation} Combining with \eqref{eq19}, the collision probability can be further expressed by
\begin{equation}
\label{eq21}
\psi_s(\Gamma)=\frac{\theta_{(1,0)}(1-\exp(-\alpha))}{(1-\phi_s)\exp(-\alpha)}.
\end{equation}

\subsection{Calculation of the average AoI}

There are two ways to calculate the average AoI, the one is from the steady state distribution, and the other is by analyzing the AoI evolution.
Using the steady state distribution, the average AoI can be expressed by
\begin{equation}
\label{e22}
\bar{\Delta}(\Gamma)=\sum_{\delta}\delta\sum_{u}\theta_{(\delta,u)}.
\end{equation} By substituting Eqs. \eqref{e17}, \eqref{e18}, \eqref{e19}, \eqref{e20} and \eqref{e21} into Eq.\eqref{e22}, we can calculate the average AoI, as given in \eqref{eq22}.
\begin{equation}
\label{eq22}
\bar{\Delta}(\Gamma)=\Gamma-\frac{\left(\frac{\Gamma (\Gamma-1)}{2}-(1-\frac{\alpha+\beta}{\beta(1-\exp(-\alpha+\beta))}-\frac{\alpha+\beta}{\beta e^{-\alpha}(1-\phi_s)})\frac{\alpha \exp(-(\alpha+\beta)(\Gamma-1))}{\beta(1-\exp(-\alpha-\beta))}-\Xi\right)}{\Gamma-1+\frac{\alpha+\beta}{\beta e^{-\alpha}(1-\phi_s)}+\frac{\alpha}{(1-\exp(-\alpha+\beta))\beta}\left(1-\exp(-(\alpha+\beta)(\Gamma-1))\right)}
\end{equation}
\begin{small}
	\begin{equation}
	\begin{aligned}
	\Xi=\frac{{\left((\alpha+\beta)\exp(\alpha)+\alpha(1-\phi_s)\right)}^2}{\beta^2{(1-\phi_s)}^2}-\frac{(\beta+\alpha)\exp(\alpha)}{\beta(1-\phi_s)}+\frac{\left(\frac{2\alpha(\alpha+\beta)(\exp(a)+1-\phi_s)}{\beta^2(1-\phi_s)}-\frac{\alpha}{\beta}\right)}{\exp(\alpha+\beta)-1}+\frac{\alpha(\alpha+\beta)}{\beta^2{(\exp(\alpha+\beta)-1)}^2}
	\end{aligned}
	\end{equation}
\end{small}

\subsection{Calculation of optimal stationary policy}
Based on the above analysis of the threshold policy, given the constrained average collision probability of SIoT $\eta_s$, we can calculate the optimal stationary policy $(\mu, \pi^{*}_{\lambda_1}, \pi^{*}_{\lambda_2})$ which randomly selects one of the two threshold policy $\pi_{\Gamma_1^*}$ and $\pi_{\Gamma_2^*}$, with probability $\mu$ at each state by solving the following equation
\begin{equation}
\label{s0}
\psi_s(\Gamma)=\frac{\theta_{(1,0)}(1-\exp(-\alpha))}{(1-\phi_s)\exp(-\alpha)}=\eta_s.
\end{equation} Then we have 
\begin{equation}
\label{s1}
\Gamma_1^*=\lfloor \frac{1}{\alpha+\beta} W\left(\frac{(\alpha+\beta)\alpha}{\beta(1-e^{-(\alpha+\beta)})}\exp\left(\frac{(\alpha+\beta)\alpha}{\beta(1-e^{-(\alpha+\beta)})}\Upsilon\right)\right) -\frac{\alpha\Upsilon}{\beta(1-e^{-(\alpha+\beta)})}+1\rfloor,
\end{equation}
\begin{equation}
\label{s2}
\Gamma_2^*=\lceil \frac{1}{\alpha+\beta} W\left(\frac{(\alpha+\beta)\alpha}{\beta(1-e^{-(\alpha+\beta)})}\exp\left(\frac{(\alpha+\beta)\alpha}{\beta(1-e^{-(\alpha+\beta)})}\Upsilon\right)\right) -\frac{\alpha\Upsilon}{\beta(1-e^{-(\alpha+\beta)})}+1\rceil,
\end{equation} where 
$\Upsilon=\frac{{(\alpha+\beta)(1-\exp(-\alpha-\beta))}}{\alpha \exp(-\alpha)(1-\phi_s)}+1-\frac{(1-\exp(-\alpha))(1-\exp(-\alpha-\beta))\beta}{\eta_s(1-\phi_s)\exp(-\alpha)\alpha}$ and $W(\cdot)$ is Lambert W function\cite{dence2013brief}. 

If $\Gamma_1^*=\Gamma_2^*$, then, the optimal policy is a deterministic policy with threshold $\Gamma_1^*$, i.e., the optimal policy $\pi(\Gamma_1^*)$ for the SIoT is to sample and deliver its status update when the channel is sensed idle and current AoI is not smaller than $\Gamma_1^*$. 

If $\Gamma_1^*\ne\Gamma_2^*$, i.e., $\Gamma_1^*=\Gamma_2^*-1$, then the optimal policy is $(\mu, \pi(\Gamma_1^*),\pi(\Gamma_2^*))$. We now need to calculate the randomization probability $\mu$. Based on the previous analysis of threshold policy, we have
\begin{equation}
\label{e17v1}
\theta_{(\delta,0)}=\theta_{(1,0)}\frac{\beta+\alpha \exp(-(\alpha+\beta)(\delta-1))}{\alpha+\beta},
\end{equation}
\begin{equation}
\label{e18v1}
\theta_{(\delta,1)}=\theta_{(1,0)}\frac{\alpha-\alpha \exp(-(\alpha+\beta)(\delta-1))}{\alpha+\beta},
\end{equation} for $\delta\leq \Gamma_1^*$ under $(\mu, \pi(\Gamma_1^*),\pi(\Gamma_2^*))$. The difference between the steady state distributions under threshold policy $\pi(\Gamma_1^*)$ and randomized simple policy $(\mu, \pi(\Gamma_1^*),\pi(\Gamma_2^*))$ is the steady state probability of $\{\theta_{(\delta,u)}\}_{\delta>\Gamma_1^*}$. This is because for state $s=(\Gamma_1^*,0)$, the SIoT delivers its status update with probability 1 under threshold policy $\pi(\Gamma_1^*)$ and with probability $\mu$ under the randomized simple policy $(\mu, \pi(\Gamma_1^*),\pi(\Gamma_2^*))$. Thus, similar to the above Markov chain analysis for the threshold policy, we have
\begin{equation}
\label{e17v2}
\theta_{(\Gamma_2^*,0)}=\theta_{(1,0)}\left(\frac{\beta+\alpha \exp(-(\alpha+\beta)\Gamma_1^*)}{\alpha+\beta}-\mu\exp(-\alpha)(1-\phi_s)\frac{\beta+\alpha \exp(-(\alpha+\beta)(\Gamma_1^*-1))}{\alpha+\beta}\right),
\end{equation}
\begin{equation}
\label{e18v2}
\theta_{(\Gamma_2^*,1)}=\theta_{(1,0)}\frac{\alpha-\alpha \exp(-(\alpha+\beta)\Gamma_1^*)}{\alpha+\beta}.
\end{equation} For the states $\{(\delta,u)\}_{\delta\geq\Gamma_2^*}$, their steady state distribution is similar to that under threshold policy $\pi(\Gamma_1^*)$, because the action for states $\{(\delta,0)\}_{\delta\geq\Gamma_2^*}$ is to sample and deliver status update. Similarly, we have
\begin{equation}
\label{e19v1}
\theta_{(\delta,0)}=\theta_{(\Gamma_2^*,0)}\frac{p_{BI}}{m}\left(ab^{\delta-\Gamma_2^*}-cd^{\delta-\Gamma_2^*}\right)+\theta_{(\Gamma_2^*,1)}\frac{p_{BI}}{m}\left(b^{\delta-\Gamma_2^*}-d^{\delta-\Gamma_2^*}\right),
\end{equation}
\begin{equation}
\label{e20v1}
\theta_{(\delta,1)}=\theta_{(\Gamma_2^*,0)}\frac{p_{BI}}{m}ac\left(d^{\delta-\Gamma_2^*}-b^{\delta-\Gamma_2^*}\right)+\theta_{(\Gamma_2^*,1)}\frac{p_{BI}}{m}\left(ad^{\delta-\Gamma_2^*}-cb^{\delta-\Gamma_2^*}\right),
\end{equation} for $\delta\geq\Gamma_2^*$. The equality of \eqref{eq21} still holds. Then, by solving \eqref{s0}, we can calculate the $\theta_{(1,0)}$. Substituting the calculated $\theta_{(1,0)}$, \eqref{e17v1}, \eqref{e18v1}, \eqref{e17v2}, \eqref{e18v2}, \eqref{e19v1} and \eqref{e20v1} into $\sum_{u}\sum_{\delta}\theta_{(\delta,u)}=1$, the value of $\mu$ can then be calculated as following
\begin{small}
\begin{equation}
\mu=\left(\Gamma_1^*+\frac{1-\left(1-\exp(-\alpha)\right)/\eta_s+\alpha/\beta}{(1-\phi_s)\exp(-\alpha)}+\frac{\alpha(1-\exp\left(-(\alpha+\beta)\Gamma_1^*\right))}{\beta\left(1-\exp(-\alpha-\beta)\right)}\right)\frac{\beta}{\beta+\alpha\exp\left(-(\alpha+\beta)(\Gamma_1^*-1)\right)}.
\end{equation}
\end{small} This completes the closed-form calculation of the optimal policy given the constraint $\eta_s$ and system parameters.

\section{Throughput-Optimal Policy}
In this part, we compare the throughput-optimal policy (the policy that maximizes throughput of SIoT in CR-IoT system) with the derived age-optimal policy and analyze the average AoI performance for throughput-optimal policy in CR-IoT system. According to \cite{huang2009optimization}, the throughput-optimal policy is a randomized policy. In this policy, when the channel state is idle, the SIoT device transmits its status update with a fixed probability $p_0$ to achieve the maximum allowable collision probability $\eta_s$. The state transition diagram under the randomized policy is illustrated in Fig. \ref{fig10}. The optimal throughput of the SIoT is $p_I p_0$ and the collision probability is $p_0 p_I (1-\exp(-\alpha))$. Comparing to the derived age-optimal policy, the throughput-optimal policy only depends on the transmission probability $p_0$, while the age-optimal policy relies on the AoI threshold $\Gamma$ and randomization probability $\mu$. 

\begin{figure}[!h]
	\centerline{\includegraphics[width=0.45\textwidth]{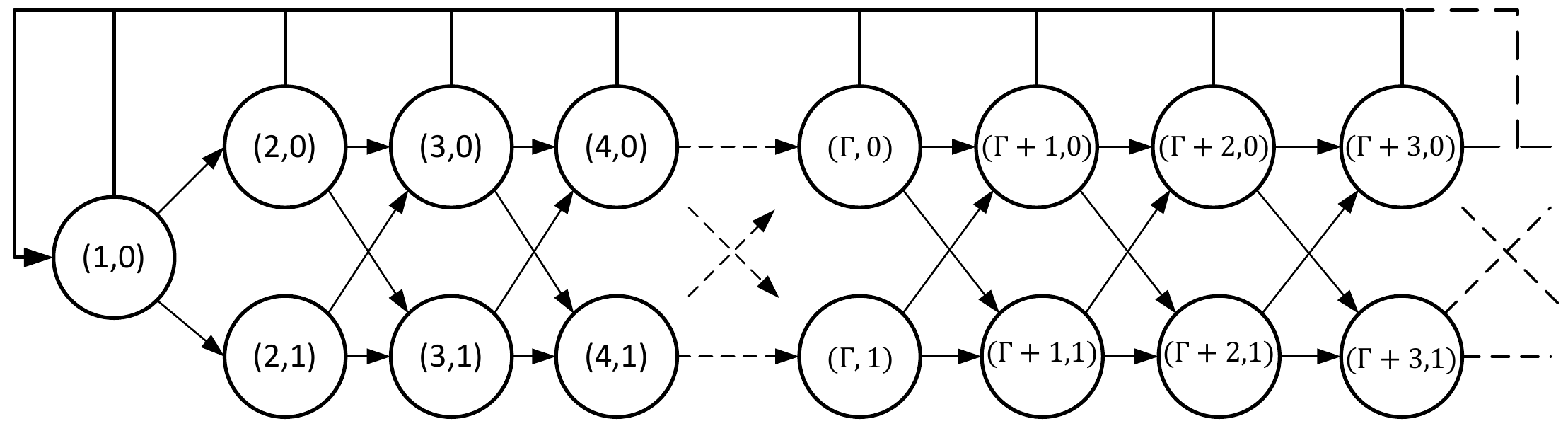}}
	\caption{Two dimensional Markov Chain under randomized policy that transmit a status update with probability $p_0$ when channel is sensed idle, where state $(i,j)$ indicates AoI is $i$, channel is idle when $j=0$ and channel is idle when $j=1$. }
	\label{fig10}
\end{figure}
We now analyze the AoI performance under the throughput-optimal policy. The probability $p_0$ can be calculated as 
\begin{equation}
\label{eq39}
p_0=\frac{\eta_s}{\frac{\beta}{\alpha+\beta}(1-\exp(-\alpha))},
\end{equation} considering the long-term average idle probability of the PU $p_I$ and the probability that PU returns during a slot if the PU is idle at the beginning of the slot. Based on the calculated probability in \eqref{eq39}, we have the following equation regarding the relationship among different states,
\begin{equation}
\left(\theta_{(\delta,0)},\theta_{(\delta,1)}\right)
\begin{pmatrix}
C& p_{IB}\\
p_{BI}& p_{BB}
\end{pmatrix}=\left(\theta_{(\delta+1,0)},\theta_{(\delta+1,1)}\right),
\end{equation} where $C=p_{II}-p_0(1-\phi_s)e^{-\alpha}$, and 
\begin{equation}
\label{eq23}
\sum_{\delta}\theta_{(\delta,0)}p_0(1-\phi_s)e^{-\alpha}=\theta_{(1,0)}.
\end{equation} Similarly, we can calculate the steady state distribution as in Section IV. Based on the distribution, the average AoI achieved by $p_0$ can be expressed as
\begin{equation}
\bar{\Delta}(p_0)=\frac{(\alpha+\beta)\exp(\alpha)}{\beta (1-\phi_s) p_0}+\frac{\alpha \exp(\alpha+\beta)}{\beta (\exp(\alpha+\beta)-1)}.
\end{equation} The detailed steps are omitted for brevity as it is similar to the analysis of the threshold policy.

\section{Numerical Results}
In this section, numerical results are presented to verify the effectiveness of the optimal policy and the accuracy of our analysis. 
\subsection{Simulation Settings}
We first generate the traffic of the PU according to the two-state homogeneous continuous-time
Markov chain model. Specifically, it consists of $10^6$ successive busy-idle cycles; the length of busy and idle periods in each cycle is continuous and follows the exponential distribution. The slot length of the SIoT is set to $1$. The SIoT device senses the channel at the beginning of the time slot. The collision probability of the PU is calculated according to \eqref{eq2} and \eqref{eq2**}. The average AoI of the SIoT is evaluated by taking the average of the AoI during the busy-idle cycles of the PU. 
\subsection{Optimal Stationary Policy for CMDP}
In this subsection, we evaluate the structure of the optimal policy. Fig. \ref{fig3} illustrates the optimal policy by Relative Value Iteration (RVI) and stochastic gradient descent for the search of $\lambda$ for a given set of $\eta_s$, $\phi_s$, $\alpha$ and $\beta$. We apply RVI on finite states ($\delta \leq 200$) to approximate the countable infinite state space according \cite[chapter 8]{sennott2009stochastic}. The optimal policy is a randomized simple policy that randomizes between two deterministic policies $\pi_{\lambda_1^*}$ and $\pi_{\lambda_2^*}$ at each state, which is displayed in Fig. \ref{fig3} and Fig. \ref{fig3*} for different collision constraints. The threshold structure of the two deterministic policies is obvious. By substituting the system parameters into \eqref{s1} and \eqref{s2}, the corresponding thresholds can be calculated and verified by comparing with the numerical results. Comparing Fig. \ref{fig3} with Fig. \ref{fig3*}, we can see that the tighter collision constraint leads to a larger AoI threshold. Furthermore, we can find that as $\lambda$ increases, the threshold of the corresponding unconstrained MDP policy increases. This is because the larger $\lambda$ indicates the higher weight of collision cost in the reward function of Lagrangian relaxed unconstrained MDP, which in turn, leads to a larger threshold to minimize the long-term average reward.
\begin{figure}[!h]
	\centering
	\begin{subfigure}[b]{0.45\textwidth}
		\centering
		\includegraphics[width=\textwidth]{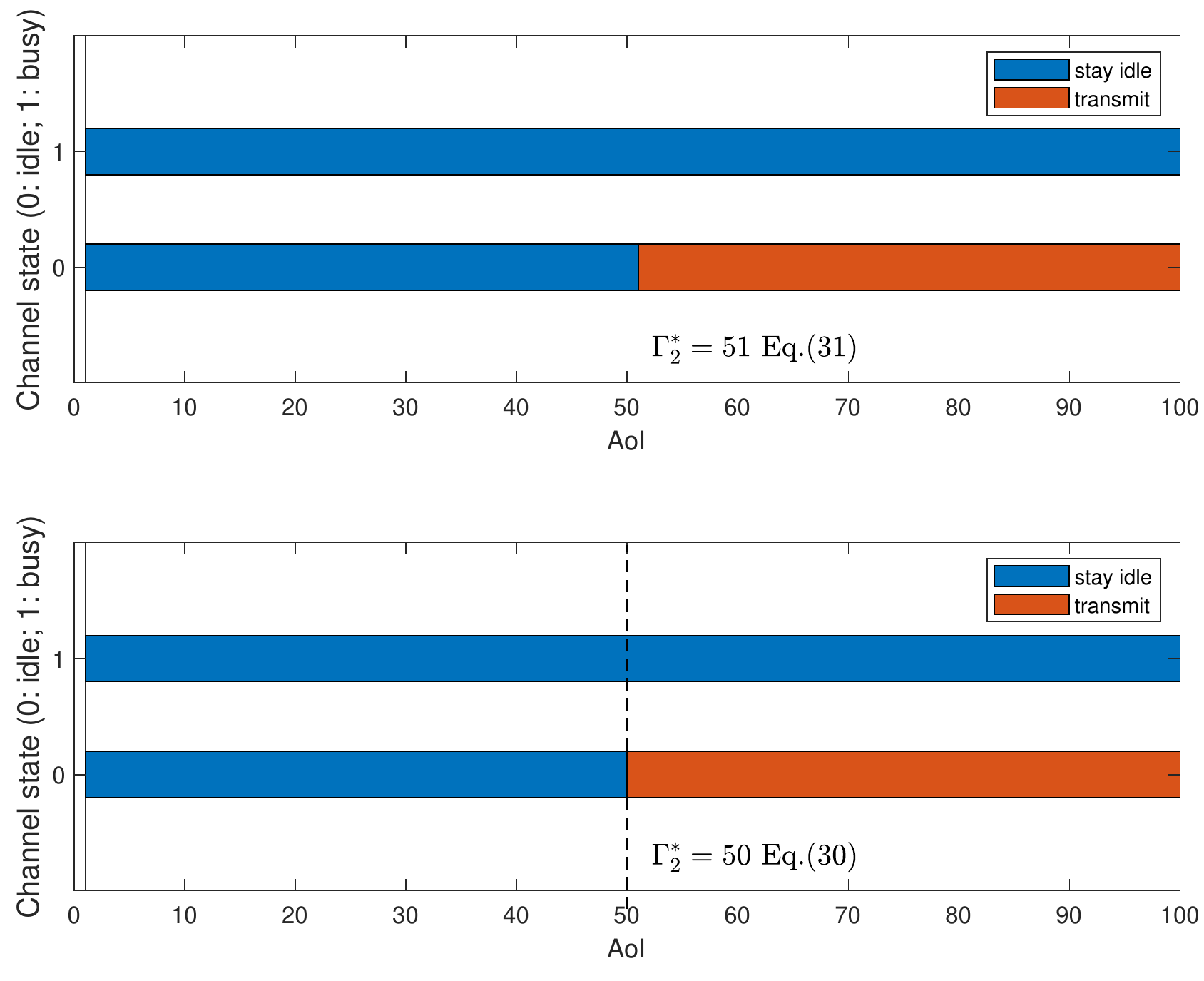}
		\caption{$\eta_s=0.0005$}
		\label{fig3}
	\end{subfigure}
	\hfill
	\begin{subfigure}[b]{0.45\textwidth}
		\centering
		\includegraphics[width=\textwidth]{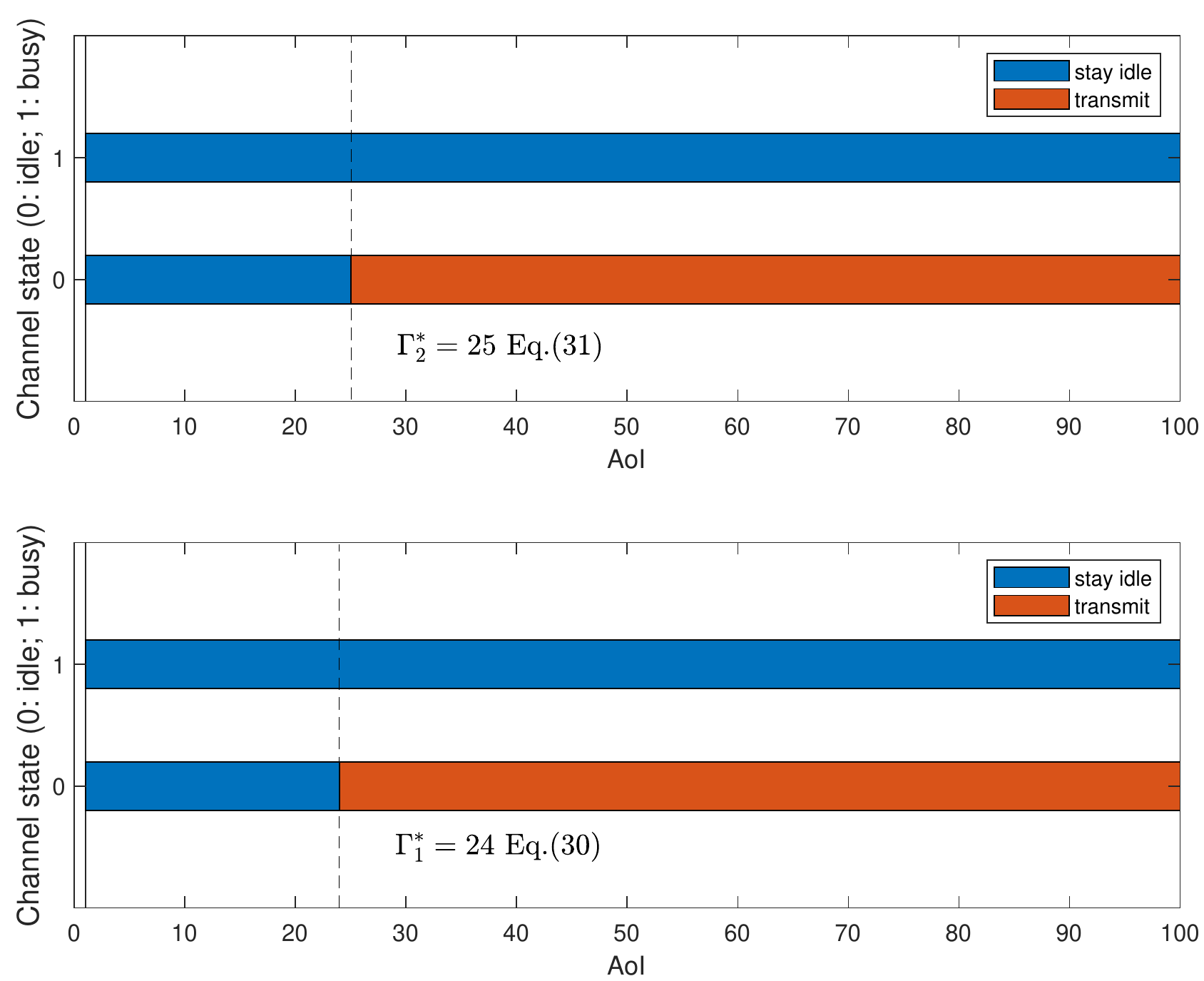}
		\caption{$\eta_s=0.001$}
		\label{fig3*}
	\end{subfigure}
	\caption{Structural deterministic policy for $\lambda_{1}^*$ (top) and $\lambda_{2}^*$ (bottom) where $\lambda_{1}^*>\lambda_{2}^*$, with $\phi_s=0.2,\alpha=0.02$ and $\beta=0.4$.}
\end{figure}

\subsection{Evaluation of Markov chain analysis}
In Fig.\ref{fig4}, we evaluate the Markov chain analysis of average AoI and collision probability. The analytical results are compared with the simulation results to verify the theoretical analysis of our system. The traffic pattern of the PU is set as $\alpha =0.02$ and $\beta=0.4$. First, we can see that our analytical results well coincide with the simulation results which verifies our Markov chain analysis of average AoI under threshold policy. Second, we can observe the trade-off between average AoI and the collision probability. Specifically, the policy with a larger threshold leads to larger average AoI and smaller collision probability, as shown in Fig.\ref{fig4}. This is reasonable because the SIoT is less likely to conduct transmission with a larger threshold, which in turn increases average AoI and decreases collision probability. Furthermore, we observe that the average AoI approximately linearly increases with respect to the threshold $\Gamma$, while the collision probability exponentially decreases.
\begin{figure}[!h]
	\centerline{\includegraphics[width=0.45\textwidth]{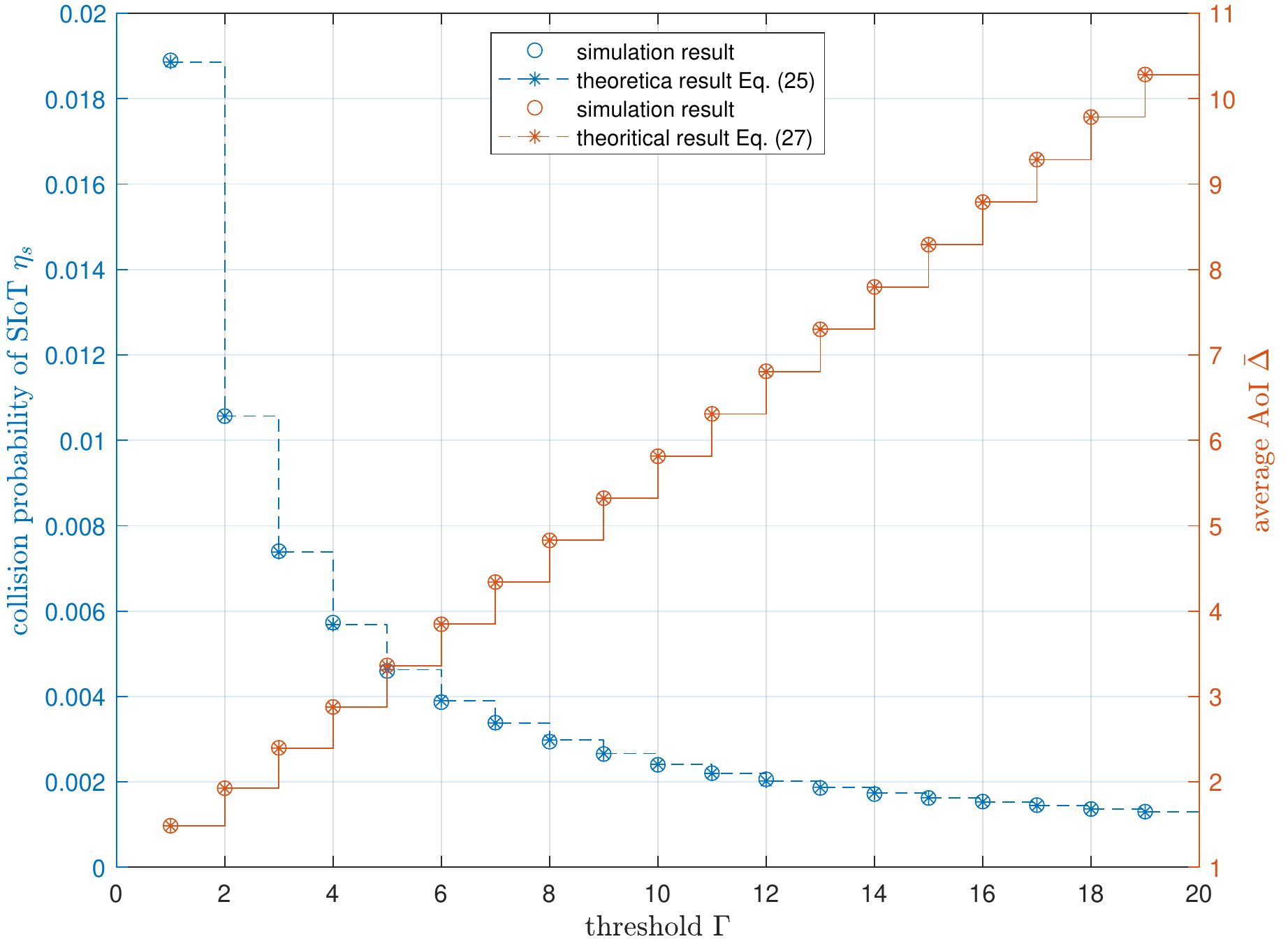}}
	\caption{Comparisons of simulation average AoI result and theoretical average AoI result of threshold policy with $\phi_s=0.2$, $\alpha=0.02$ and $\beta=0.4$.}
	\label{fig4}
\end{figure}

\subsection{Evaluation of the Optimal Policy}
\begin{figure}[!h]
	\centerline{\includegraphics[width=0.45\textwidth]{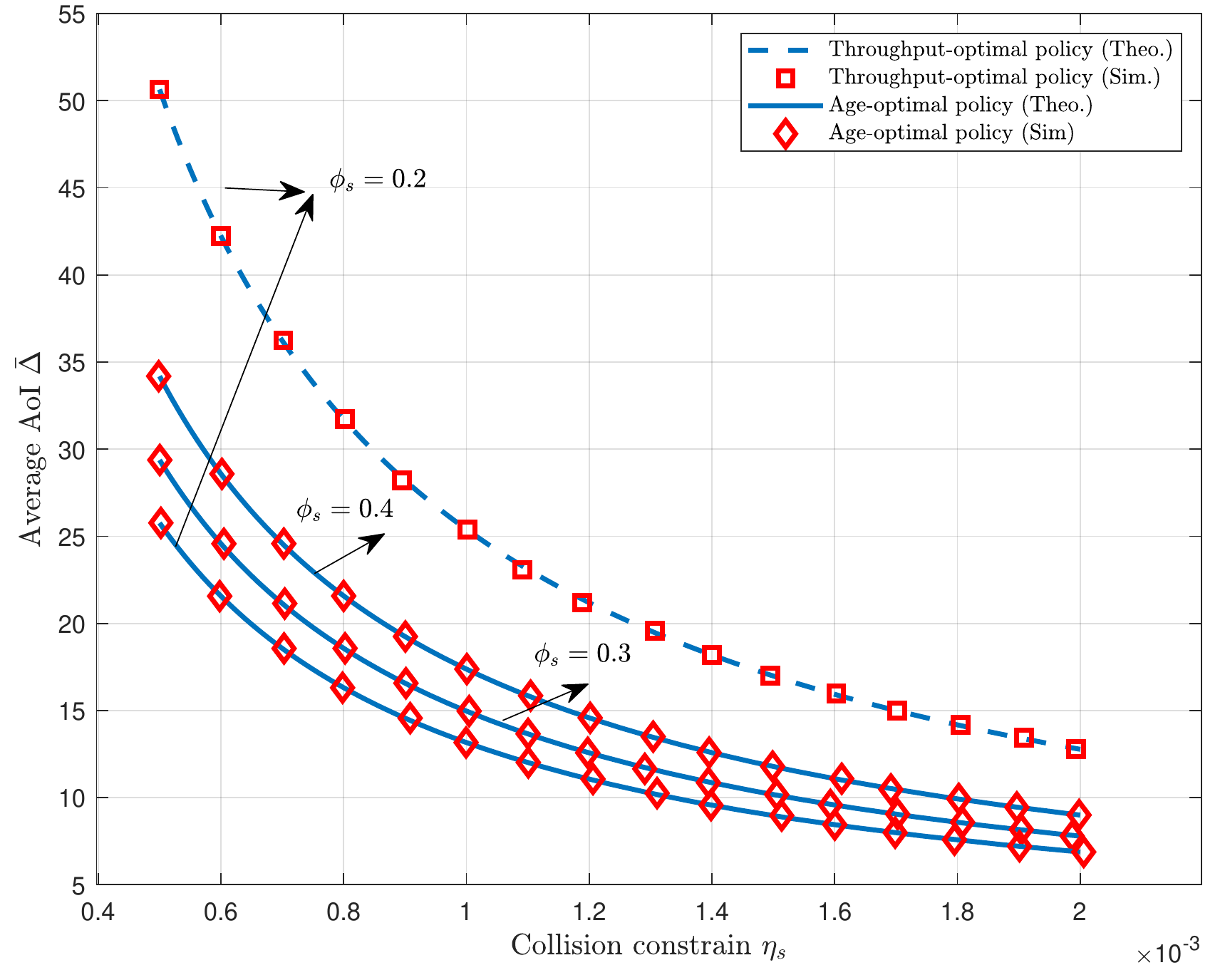}}
	\caption{Average AoI-collision probability trade-off for $\alpha=0.02$ and $\beta=0.4$.}
	\label{fig6}
\end{figure}
In Fig.\ref{fig6}, we compare performance of the age-optimal policy with different system parameters with that of the throughput-optimal policy with respect to different collision probability constraints. 
Compared with the calculated age-optimal policy, the throughput-optimal policy is simple and easy to implement at the cost of larger average AoI, as illustrated in Fig. \ref{fig6}. Specifically, the throughput-optimal policy does not rely on updating feedback from the destination to the SIoT but the calculated transmission probability. As shown in Fig. \ref{fig6}, the age-optimal policy achieves lower average AoI, comparing to the throughput-optimal policy even when suffering from a larger outage probability. This also confirms that the throughput optimization is different from AoI optimization. Moreover, for a fixed collision probability, it is obvious that larger outage probability leads to larger average AoI. In addition, we can find that when the collision constraint is tighter, the effect of outage probability on the average AoI is more significant. 

\subsection{Effect of System Parameters}
We evaluate the effect of average idle probability $p_{I}$ on the optimal average AoI by fixing average idle state length $\alpha^{-1}$ while changing the average busy state length $\beta^{-1}$ in Fig. \ref{fig5} and Fig. \ref{fig7}. The difference between Fig. \ref{fig5} and Fig. \ref{fig7} lies in the collision constrains. The collision constraints in Fig. \ref{fig5} are from the PU's perspective $\eta_p$ while that in Fig. \ref{fig7} is from the SIoT's perspective $\eta_s$. Both figures show that the larger average idle probability leads to the smaller average AoI as there is more white space in the frequency band for the SIoT to make use of (low utilization). Similarly, it is obvious that as the collision constraint of the PU and that of SU increase, the influence on the optimal average AoI is less significant. This is due to the relationship between optimal average AoI and the collision constraint as in Fig. \ref{fig6}, where the average AoI decreases at a slower speed when the collision constraint increases. Actually, for each data point with the same $\eta_s$ constraint in Fig. \ref{fig7}, the corresponding collision probability from PU's perspective is different due to the changes in $\beta$ and the relationship between $\eta_p$ and $\eta_s$ in \eqref{eq2**}. When $\eta_s$ is fixed, as idle rate $\beta$ increases, although there is more white space in PU's frequency for SIoT to use, the corresponding $\eta_p$ decreases and it further poses a tighter constraint for PU protection. These two effects jointly lead to a slower decreasing rate.

\begin{figure}[!h]
	\centering
	\begin{subfigure}[b]{0.45\textwidth}
		\centering
		\includegraphics[width=\textwidth]{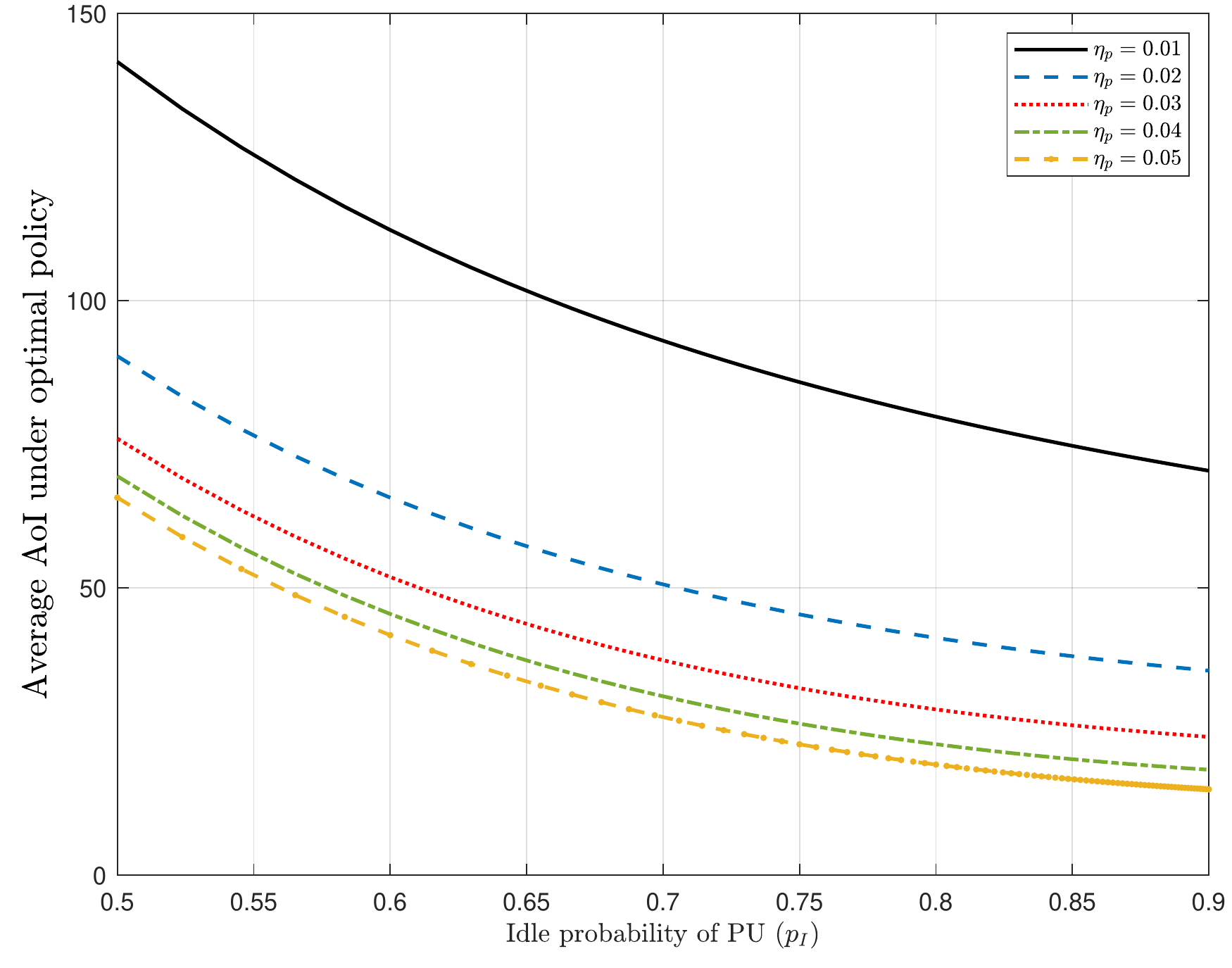}
		\caption{Different collision probability of PU $\eta_p$}
		\label{fig5}
	\end{subfigure}
	\hfill
	\begin{subfigure}[b]{0.45\textwidth}
		\centering
		\includegraphics[width=\textwidth]{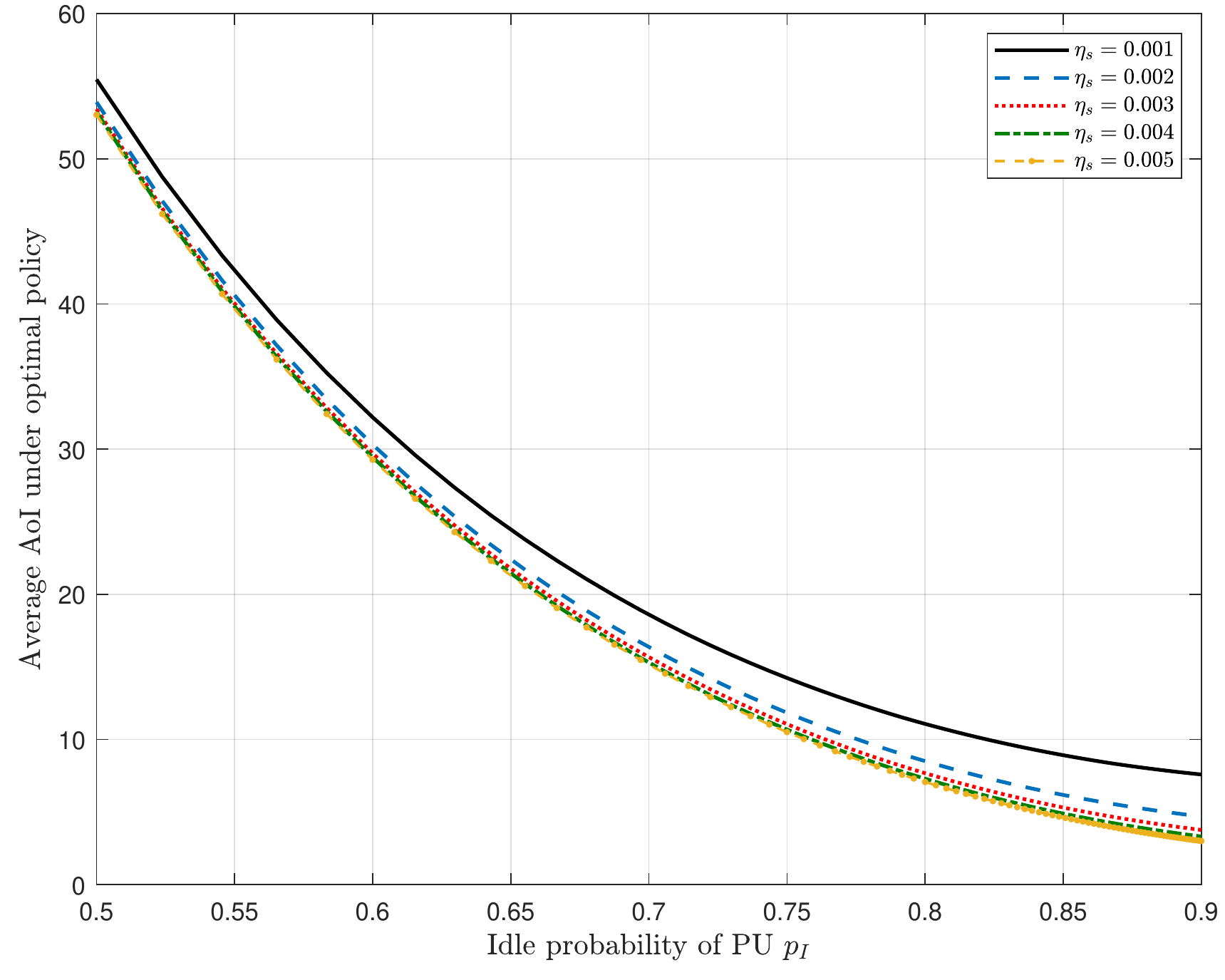}
		\caption{Different collision probability of SIoT $\eta_s$}
		\label{fig7}
	\end{subfigure}
	\caption{Idle probability $p_I=\frac{\beta}{\alpha+\beta}$ vs average AoI under optimal policy $\phi_s=0.2$, $\alpha=0.01$ with different collision probability of PU $\eta_p$ and with different collision probability of SIoT $\eta_s$.}
\end{figure}
To analyze how PU's activity frequency (busy-idle cycle length) influences the average AoI of SIoT under the collision constraint $\eta_p$, we fix the average idle probability $p_I$ and change $\alpha$ and $\beta$ accordingly. The result is depicted in Fig. \ref{fig8}. We can observe the tradeoff between the PU’s activity frequency and the average AoI performance of the SIoT for a given collision constraint of the PU.  Specifically, as the PU's activity frequency increases (smaller average idle state length $\alpha^{-1}$ and busy state length $\beta^{-1}$), the optimal average AoI of the SIoT decreases at first, and then increases gradually. The reason behind this observation is that if the PU's activity frequency is relatively low, i.e., long busy-idle cycle, then during the busy state of the PU, the AoI will keep increasing. Although in the idle state, the larger AoI effect can be alleviated to some extent, the negative effect on the average AoI during the busy period is dominant. As the PU's activity frequency increases, during the same period, it is more likely for the SIoT to cause a collision to the PU if the SIoT does not change the threshold of its transmission. This is because there will be more busy-idle cycles. Thus, the threshold of the age-optimal policy will increase accordingly to satisfy the collision constraint of PU. The larger threshold will lead to larger average AoI. These two effects together lead to the average AoI performance in as shown in Fig. \ref{fig8}. This tradeoff offers a way for the SIoT to choose the appropriate frequency band based on the PU activity in systems with multiple SIoT devices and multiple PU system. Take the special case of systems with two SIoT devices and two PUs as an example. Table \ref{tab1} givens the optimal average AoI performance of SIoT device $a$ and SIoT device $b$ under channels of different PUs. If these two devices need to select from No. $1$ and $2$ channels, then, the optimal choice to minimize the total average AoI of these two devices will be: SIoT device $a$ selects No. $2$ channel and SIoT device $b$ selects No. $1$ channel. However, when these two devices select from No. $1$ and $3$ channels, the optimal choice will be for SIoT device $a$ to select No. $1$ channel and for SIoT device $b$ to select No. $3$ channel.
\begin{figure}[!h]
	\centerline{\includegraphics[width=0.45\textwidth]{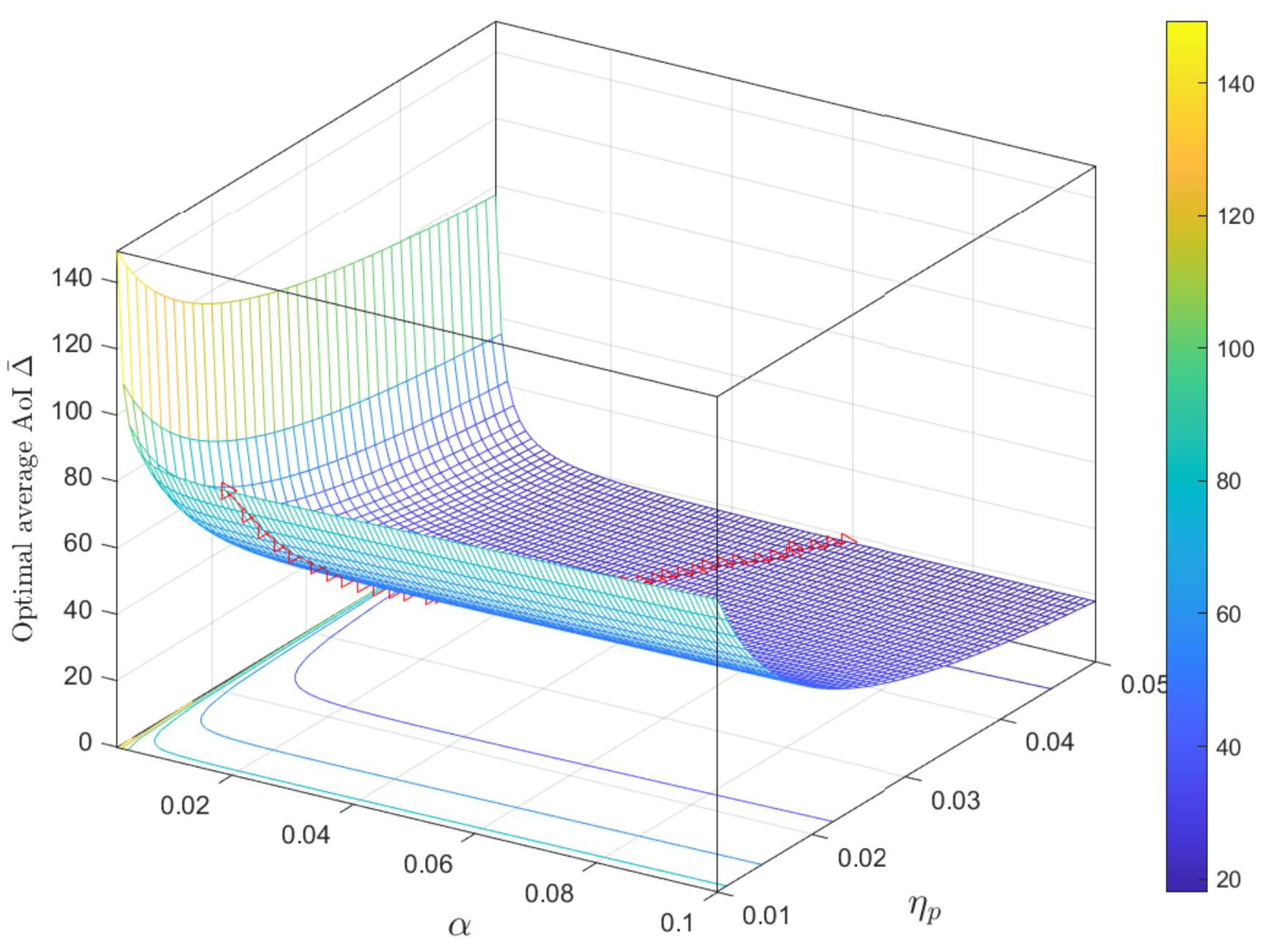}}
	\caption{Busy rate $\alpha$ vs average AoI under optimal policy  $p_I=0.75$. The red triangles are the point with optimal average AoI for fixed collision constraint of PU $\eta_p$.}
	\label{fig8}
\end{figure}

\begin{table}[!h]
	\renewcommand{\arraystretch}{1.3}
	\caption{Optimal AoI performance under channels with different PUs.}		
	\label{tab1}
	\centering
		\begin{tabular}{c*{4}{|c}}
		\toprule
		No. of PUs $p_I=0.75$  & 1& 2 &3 &4 \\
		\cline{1-1}
		 $\alpha$ & $0.002$ & $0.01$ &$0.002$ &  $0.01$\\
		 $\eta_p$ & $0.01$ &$0.01$ &$0.05$ &   $0.05$\\ 
		 \hline
		SIoT a $\phi_s=0.2$  & $109.90$ &  $ 85.82$ &$55.44$ & $22.77$ \\
	    SIoT b $\phi_s=0.3$  & $120.20$ & $97.60$ & $ 57.34$ & $24.86$  \\
		\bottomrule
	\end{tabular}
\end{table}


\section{Conclusions}
In this paper, we have studied the optimal policy that minimizes the long-term average Age of Information (AoI) for a cognitive radio-based IoT monitoring system, where an IoT device as a secondary user (SIoT)  seeks and exploits the spectrum opportunities of licensed band vacated by its primary user (PU) to deliver status updates timely with invisible effects to the licensed operation. We formulated the AoI minimization problem of the SIoT as a constrained Markov decision process (CMDP) problem. We proved that the optimal stationary policy is a randomized simple policy that randomizes at each state between two deterministic policies with a fixed probability and the deterministic policies have a threshold structure. We have derived the analytical expressions of the long-term average AoI performance and collision probability of deterministic threshold-structured policy, which offers a simple and effective way to calculate the age-optimal policy for the SIoT. We also considered the throughput-optimal policy for throughput maximization as a benchmark and derived the AoI performance under the throughput-optimal policy in the considered system. Numerical simulations showed the superior performance of the age-optimal policy comparing to the throughput-optimal policy and illustrated the difference between throughput maximization and AoI minimization. We also unveiled the impacts of various system parameters on the corresponding optimal policy and the resultant average AoI.

\begin{appendices}
	\section{Proof of Theorem \ref{T1}}
	\label{A1}
	For the case $D(\pi^{*}_{\lambda=0}|s_0) \leq \eta_s$, $\pi^{*}_{\lambda=0}$ will conduct transmissions for the SIoT if the channel is sensed idle, which is straightforward. 
	We then prove the rest part of the theorem by verifying that the Assumptions 1-5 in \cite{sennott1993constrained} hold. According to the definition of reward and cost function, Assumption 1 in \cite{sennott1993constrained} holds as in \eqref{e6}. Thus, we need to verify the rest four assumptions in \cite{sennott1993constrained}, which can be summarized in the following condition,
	\begin{itemize}
		\item There exists a stationary policy $e$ that induces an irreducible positive recurrent Markov chain on $S$ with a single positive recurrent class, such that $C(e|s_0)$ and $D(e|s_0)$ are finite, in particular  $D(e|s_0) <\eta_s$.
	\end{itemize}

We provide one of the stationary policies that satisfy the condition above for verification. A typical policy is a randomized stationary policy $\pi$ where the SIoT transmits a status update with a fixed probability $p$ if the channel is sensed idle. This policy is also the throughput-optimal policy with detailed analysis provided in Section V. The throughput-optimal policy meets the above condition and thus the Assumptions 1-5 in \cite{sennott1993constrained} hold. According to \cite[Theorem $2.5$]{sennott1993constrained}, we can prove the the case $D(\pi^{*}_{\lambda=0}|s_0) > \eta_s$ in Theorem \ref{T1}. This completes the proof.
	\section{Proof of Theorem \ref{T2}}
	\label{A2}
The threshold policy is actually the same as the monotone nondecreasing policy. To prove the monotonicity of the optimal policy of the unconstrained MDP problem, we verify that the following four conditions \cite[Theorem~8.11.3]{b5} hold.
\begin{itemize}
	\item [a)] $r(s,a)$ is nondecreasing in $s$ for all $a\in A$;
	\item [b)] $q(k|s,a)=\sum_{j=k}^{\infty}p(j|s,a)$ is nondecreasing in $s$ for all $k\in S$ and $a\in A$;
	\item [c)] $r(s,a)$ is a subadditive function on $S\times A$ and
	\item [d)] $q(k|s,a)$ is a subadditive function on $S\times A$ for all $k\in S$.
\end{itemize}
We only need to consider the case that the sensing result is idle, $u=0$, because when the channel is sensed busy, i.e.,$u=1$, conducting transmission will not lead to AoI drop but cause collision to the PU. To verify these conditions, we first order the state by $\delta$, i.e., $s^+\geq s^-$ if $\delta^+\geq \delta^-$ where $s^+=(\delta^+,\cdot)$ and $s^-=(\delta^-,\cdot)$. The one-step reward function of unconstrained MDP is 
\begin{equation}
\label{req1}
r(s,a)=\delta+\lambda (au+a(1-u)(1-e^{-\alpha})).
\end{equation} It is obvious that the condition a) is satisfied. According to the transition probabilities in \eqref{e3}, if the current state $s=(\delta,u)$, the next possible states are $s_1=(\delta+1,\cdot)$ (including $(\delta+1,0)$ and $s_2=(\delta+1,1)$) and $s_2=(1,0)$.  Based on \eqref{e3}, we have
\begin{equation}
\label{req2}
q(k|s,a=0)=\left\{
\begin{array}{rcl}
0,& &\text{if }  k > s_1 \\
1,& &\text{otherwise.}  \\
\end{array}
\right.
\end{equation}
\begin{equation}
\label{req3}
q(k|s,a=1)=\left\{
\begin{array}{rcl}
0,& &\text{if }  k > s_1\\
1-e^{-\alpha}(1-\phi_s),& &\text{if }  s_1\geq k>s_2 \\
1,& & \text{if } k=s_2
\end{array}
\right.
\end{equation} Thus, condition b) is immediate. To verify the rest two conditions, we give the definition of subadditive in the following
\begin{definition}
	\label{d0}
	(Subadditive\cite{b5}) A multivariable function $Q(\delta,u,a): \mathcal{D} \times \mathcal{U} \times \mathcal{A} \rightarrow R$ is subadditive in $(\delta,u,a)$ for fixed parameter $u \in U$, if for all $\delta^{+}\geq\delta^{-}$ and $a^{+}\geq a^{-}$,
	\begin{equation}
	\label{e8}
	Q(\delta^{+},a^{+};u)+ Q(\delta^{-},a^{-};u) \leq Q(\delta^{+},a^{-};u)+ Q(\delta^{-},a^{+};u)
	\end{equation}holds.
\end{definition} 
According to \eqref{req1}, condition c) follows. For the last condition, we verify whether
\begin{equation}
q(k|s^+,a^+)+q(k|s^-,a^-)\leq q(k|s^+,a^-)+q(k|s^-,a^+),
\end{equation} with $s^+=(\delta^+,u)$ and $s^-=(\delta^-,u)$ where $\delta^+\geq \delta^-$ and $a^+\geq a^-$. The equality holds for the case where $u=1$. So, we focus on the case $u=0$. If $k=s_2$ or $k>s_1$, the equality holds as well; otherwise, the left part is $1-e^{\alpha}(1-\phi_s)$ and right part is $1$. Consequently, condition d) holds. As all these four conditions hold, the optimal policy is monotone nondecreasing is $\delta$. This completes the proof.


\section{Proof of Lemma \ref{l2} and Lemma \ref{l3}}
\label{A3}

According to \eqref{eq1} and \eqref{eq16}, we have
\begin{equation}
\label{ea16}
\left(\theta_{(\delta+1,0)},\theta_{(\delta+1,1)}\right)=\left(\theta_{(1,0)},\theta_{(1,1)}\right)
\Sigma^{\delta-1}.
\end{equation}
According to \cite[E.q. (16)-(60)]{papoulis2002probability}, we have
\begin{equation}
\label{ea17}
\Sigma^{t}=\frac{1}{\alpha+\beta}\begin{pmatrix}
\beta+\alpha e^{-(\alpha+\beta)t}&\alpha-\alpha e^{-(\alpha+\beta)t}\\
\beta-\beta e^{-(\alpha+\beta)t}& \alpha+\beta e^{-(\alpha+\beta)t}
\end{pmatrix}.
\end{equation} By substituting \eqref{ea17} and $\theta_{(1,1)}=0$ into \eqref{ea16}, Lemma \ref{l2} can be obtained. 

As for Lemma \ref{l3}, according to \eqref{eq18}, we have
\begin{equation}
\label{ea18}
\left(\theta_{(\delta,0)},\theta_{(\delta,1)}\right)=\left(\theta_{(\Gamma,0)},\theta_{(\Gamma,1)}\right)
\begin{pmatrix}
A& p_{IB}\\
p_{BI}& p_{BB}
\end{pmatrix}^{\delta-\Gamma}.
\end{equation}By doing eigenvalue decomposition, we have 
\begin{equation}
\label{ea19}
\begin{pmatrix}
A& p_{IB}\\
p_{BI}& 1-p_{BI}
\end{pmatrix}^{\delta-\Gamma}=
\begin{pmatrix}
c& a\\
1& 1
\end{pmatrix}
	\begin{pmatrix}
d& 0\\
0& b
\end{pmatrix}^{\delta-\Gamma}\left(\frac{p_{BI}}{m}
\begin{pmatrix}
-1& a\\
1& -c
\end{pmatrix}\right)
\end{equation}
By substituting \eqref{ea19}, $\theta_{(\Gamma,1)}$ and $\theta_{(\Gamma,0)}$ into \eqref{ea18}, Lemma \ref{l3} can be proofed. This completes the proof.
\end{appendices}
\bibliography{ref}
\bibliographystyle{IEEEtran}
\end{document}